%% file: HinrichsHiroshima-NelsonErgodicity-arXiv.tex
\documentclass[a4paper]{amsart}\usepackage[a4paper,margin=2.5cm,includefoot]{geometry}
\input{header}
\usepackage{hiroshima2,mathrsfs}

\title[Ergodicity of Renormalized Nelson Semigroups]{On the Ergodicity of Renormalized\\ Translation-Invariant Nelson-Type Semigroups}
\author{Benjamin Hinrichs}
\address{Benjamin Hinrichs, Universit\"at Paderborn, Institut f\"ur Mathematik, Institut f\"ur Photonische Quantensysteme, Warburger Str. 100, 33098 Paderborn, Germany}
\email{benjamin.hinrichs@math.upb.de}
\author{Fumio Hiroshima}
\address{Fumio Hiroshima, Faculty of Mathematics, Kyushu University
	819-0395 744 Nishi-ku Motooka, Fukuoka, Japan}
\email{hiroshima.fumio.965@m.kyushu-u.ac.jp}

\subjclass[2020]{Primary 81S40; Secondary 47D08.}

\usepackage{xcolor}\usepackage{cancel}

\newcommand{\nr}{\mathsf{nr}}
\newcommand{\sr}{\mathsf{sr}}
\newcommand{\FP}{{\mathcal F}\!_+}
\newcommand{\FBP}{\bar{\mathcal F}\!_+}
\newcommand{\dGn}[1]{{\rm d}\Gamma\!_{#1}}

\begin{document}
\setlength{\baselineskip}{15pt}
\begin{abstract} 
	\noindent
	We present a simple functional integration based proof that the semigroups generated by the ultraviolet-renormalized translation-invariant non- and semi-relativistic Nelson Hamiltonians are po\-si\-ti\-vi\-ty improving (and hence ergodic) with respect to the Fr\"ohlich cone for arbitrary values of the total momentum.
	Our argument simplifies known proofs for ergodicity and the result is new in
	the semi-relativistic case.
\end{abstract}

\maketitle

\section{Introduction}
A Hilbert space operator is called {\em positivity improving} if it maps the non-zero elements of a self-dual convex cone (of `positive' elements) to its strictly positive elements, i.e., the elements having strictly positive inner product with every element of the cone. It is called {\em ergodic} if its $n$-th power is positivity improving for some $n\in{\mathbb N}$. Here ${\mathbb N}=\{1,2,\ldots\}$.  This property has significant applications in spectral theory, since it implies that any maximal eigenvalue of the operator is non-degenerate and the corresponding eigenspace is spanned by a unique strictly positive vector, by the Perron--Frobenius--Faris theorem \cite{Gross.1971,Faris.1972}.
In the case of strongly continuous selfadjoint semigroups the notions positivity improving and ergodic coincide \cite{Simon.1973} and are hence often used interchangeably, as in this paper.

In the present article, we study the ergodicity of semigroups obtained as the limit of ergodic semigroups.
Since strict inequalities do not carry over in limits, it is natural that concrete properties have to be utilized to prove ergodicity in this case. 
Our proof hereby makes use of a functional integral representation of the semigroup, explicitly given in a Feynman--Kac type formula.

The models we are interested in stem from quantum field theory. In this field of research, ergodicity arguments and hence the non-degeneracy of ground states has been exploited in many works on spectral theory, see for example \cite{BachFroehlichSigal.1998a,Faris.1972,Frohlich.1974,GlimmJaffe.1970b,Hiroshima.2000b,Moller.2005}.

Explicitly, we are
concerned with the three-dimensional non-relativistic and the two-dimensional semi-relativistic 
Nelson models from quantum field theory.
They describe a spinless non- and semi-relativistic particle, respectively, linearly coupled to a bosonic quantum field.
Due to the ultraviolet divergence of the quantum field, the corresponding operators can only be defined as a selfadjoint lower-semibounded operator via a renormalization procedure. This has been done in \cite{Nelson.1964,Cannon.1971,Ammari.2000,GubinelliHiroshimaLorinczi.2014,MatteMoller.2018,GriesemerWuensch.2018,LampartSchmidt.2019} in the non-relativistic and in \cite{Sloan.1974,Schmidt.2019,HinrichsMatte.2022,HinrichsMatte.2023} in the semi-relativistic case.
In the absence of an external potential, the Hamiltonian decomposes with respect to the total momentum and one can study a family of Hamiltonians  $H_\nr(P)$ with $P\in{\mathbb R}^3$ and $H_\sr(P)$ with $P\in{\mathbb R}^2$ in the non- and semi-relativistic case, respectively, which act solely on the Hilbert space of the quantum field.
Both models (with and without ultraviolet regularization) are an active field of research in spectral and scattering theory. For example, the non-relativistic model has been studied in \cite{Cannon.1971,Frohlich.1973,Frohlich.1974,Pizzo.2003,Pizzo.2005,Moller.2005,Moller.2006,DybalskiPizzo.2014,DamHinrichs.2021,HiroshimaMatte.2019,HaslerHinrichsSiebert.2023,BeaudDybalskiGraf.2021,BachmannDeckertPizzo.2012} and the semi-relativistic model in \cite{Sloan.1974,Gross.1973,DeckertPizzo.2014,Dam.2018,HaslerHinrichsSiebert.2023}.

As far as ergodicity goes, we want to study a natural selfdual convex cone in the Hilbert space of the quantum field called {\em Fr\"ohlich cone}. In the non-relativistic case and in presence of an ultraviolet regularization, ergodicity of the semigroup in Nelson's model for arbitrary total momentum $P$ has been proven and utilized in \cite{Frohlich.1973,Moller.2005}. Further, Fr\"ohlich conjectured in \cite{Frohlich.1973} that the ultraviolet-renormalized Hamiltonian already known from \cite{Nelson.1964,Cannon.1971} at the time, also generates an ergodic semigroup. Nevertheless, rigorous proofs of this fact have only recently been given in the articles \cite{Miyao.2018,Miyao.2019,Lampart.2020} by Miyao and Lampart, using different methods of proof. Inspired by this progress, we give a simple novel proof using a functional integration method building upon the Feynman--Kac formula for the ultraviolet-renormalized model from \cite{MatteMoller.2018}.
As a byproduct, we also prove ergodicity of the Fr\"ohlich polaron \cite{Frohlich.1954} with respect to the Fr\"ohlich cone for arbitrary total momentum along the way, cf. \cref{subsec:formren}.
Our method directly carries over to the semi-relativistic model (and in fact some technical aspects of our proof become significantly simpler in this case), for which the ergodicity (with respect to a different selfdual convex cone) had to the authors knowledge up to date only been proven in the case $P=0$ \cite{Sloan.1974}.
As Lampart points out in \cite{Lampart.2020}, a combination of the techniques in \cite{Schmidt.2019,Lampart.2020} might provide an alternative approach to our result in the semi-relativistic case.

In fact, our general positivity results presented in \cref{sec:positivity} bear the potential to be applied to a broad class of translation-invariant models, see \cite{Lampart.2023} for an abstract result on renormalization of such models. Especially, in combination with a suitable Feynman--Kac formula, the authors hope that the result will contribute to proving ergodicity for the semigroup of the Bogoliubov--Fr\"ohlich Hamiltonian renormalized in \cite{Lampart.2019} and in the translation-invariant case further studied in \cite{HinrichsLampart.2023}.

\section{Main Results on the Nelson Model}\label{sec:results}
In this \lcnamecref{sec:results}, we first briefly introduce Fock space calculus (\cref{subsec:Fock}), the ultraviolet renormalized  non- and semi-relativistic Nelson models (\cref{subsec:models}) and the Fr\"ohlich cone (\cref{subsec:cone}). Then, we state and prove our main result on the corresponding semigroups (\cref{subsec:erg}). The proof combines the positivity results presented in the subsequent \cref{sec:positivity} with the Feynman--Kac formulas from \cite{Hiroshima.2015,MatteMoller.2018,HinrichsMatte.2023}, collected in \cref{subsec:FK}.

\subsection{Fock Space Calculus}\label{subsec:Fock}
We give a brief overview of the Fock space objects used in this article. For more details, see for example \cite{Parthasarathy.1992,Arai.2018}.

For now let $d\in{\mathbb N}$ be arbitrary but fixed.
Given any measurable set $\Omega\subset{\mathbb R}^d$ and using the convention that $L^2(\emptyset)$ is the trivial vector space, let
\begin{align}
	{\mathcal F}(\Omega) &\coloneqq \bigoplus_{n=0}^\infty {\mathcal F}^{(n)}(\Omega),\\
	\nonumber& {\mathcal F}^{(0)}(\Omega) \coloneqq {\mathbb C}\\
	\nonumber& {\mathcal F}^{(n)}(\Omega) \coloneqq L^2_{\mathsf{sym}}(\Omega^n) \\\nonumber & \qquad\qquad = \{f\in L^2(\Omega^n)| f(k_{\pi(1)},\ldots,k_{\pi(n)})=f(k_1,\ldots,k_n),\pi\in{\mathcal S}_n,\ \text{a.e.}\ k\in\Omega^n\}
\end{align}
denote the boson Fock space over $L^2(\Omega)$, where ${\mathcal S}_n$ is the symmetric group of degree $n$.
Throughout this article, we will denote by ${\rm P}\!_n$ the orthogonal projection onto ${\mathcal F}^{(n)}(\Omega)$.

For measurable $\Theta\subset\Omega$ and recalling the unitary equivalence $L^2(\Omega)\cong L^2(\Theta)\oplus L^2(\Omega\setminus\Theta)$, it is easy to check that there are natural unitaries
\begin{align}\label{eq:Fockres}
 {\mathcal F}(\Omega) \cong {\mathcal F}(\Theta)\otimes {\mathcal F}(\Omega\setminus\Theta) \cong {\mathcal F}(\Theta)\oplus {\mathcal F}(\Theta)^\perp, 
 \end{align}
where 
\[\qquad {\mathcal F}(\Theta)^\perp = {\mathcal F}(\Theta)\otimes \bigoplus_{n=1}^\infty
{\mathcal F}^{(n)}(\Omega\setminus\Theta)\]
and 
 on the most right hand side of \cref{eq:Fockres} we used the identification
\begin{align*}
{\mathcal F}(\Theta)\cong 
{\mathcal F}(\Theta)\otimes {\mathcal F}^{(0)}(\Omega\setminus\Theta)
=
{\mathcal F}(\Theta)\otimes {\mathbb C}.
\end{align*}
We denote the composition of this unitary with the orthogonal projection onto ${\mathcal F}(\Theta)$ in the direct sum on the right hand side of \cref{eq:Fockres} by ${\rm Q}_\Theta^\Omega$, i.e., 
\begin{subequations}\label{Qexplicit}
\begin{align}
\label{24}
&({\rm Q}_\Theta^\Omega\Psi)^{(n)}=\Psi^{(n)}\lceil_{\Theta^n},\\
\label{25}
&((\one-{\rm Q}_\Theta^\Omega)\Psi)^{(n)}=\sum_{k=0}^{n-1} 
\Psi^{(k)}\lceil_{\Theta^k}\otimes \Psi^{(n-k)}\lceil_{(\Omega\setminus \Theta)^{n-k}}.
\end{align}
\end{subequations}
Given a multiplication operator $m:\Omega\to{\mathbb C}$ and $n\in{\mathbb N}_0\coloneqq{\mathbb N}\cup\{0\}$, we define the multiplication operator $\dGn n(m)$ acting on ${\mathcal F}^{(n)}(\Omega)$ as 
\begin{align}
	\dGn0(m) \coloneqq 0  , \qquad \dGn n(m)(k_1,\ldots,k_n) \coloneqq \sum_{i=1}^n m(k_i)
\end{align}
and the second quantization of $m$ as
\begin{align}\label{def:dG}
	{\rm d}\Gamma(m) \coloneqq \bigoplus_{n=0}^\infty \dGn n(m).
\end{align}
If $m=(m_1,\ldots,m_k)$ is a vector of multiplication operators, then we consider its second quantization ${\rm d}\Gamma(m)=({\rm d}\Gamma(m_1),\ldots,{\rm d}\Gamma(m_k))$ as a vector of operators as well.

Given a test function $f\in L^2(\Omega)$, we define the annihilation operator $a(f)$ by
\begin{align}
	&a_n(f):{\mathcal F}^{(n+1)}(\Omega)\to {\mathcal F}^{(n)}(\Omega), \qquad a_n(f)\psi (k_1,\ldots,k_n) \coloneqq \sqrt{n+1}\int_{{\mathbb R}^d} \overline{f(k)}\psi(k,k_1,\ldots,k_n){\rm d}k,\nonumber\\
	&{\mathscr D}(a(f)) \coloneqq \{\psi\in{\mathcal F}^{(n)}| {\textstyle\sum_{n=0}^\infty}
	\|a_n(f){\rm P}\!_{n+1}\psi\|^2 <\infty \}, \qquad {\rm P}\!_n a(f)\psi \coloneqq a_n(f){\rm P}\!_{n+1}\psi.
	\label{def:ann}
\end{align}
It is a closed densely defined operator on ${\mathcal F}(\Omega)$. Its adjoint is the creation operator $\ad(f)=a(f)^*$, which for $f_1,\ldots,f_m\in L^2(\Omega)$ satisfies
\begin{align}
	\begin{aligned}
		&({\rm P}\!_n\ad(f_1)\cdots \ad(f_m) \psi )(k_1,\ldots,k_n) 
		\\&\quad
		=\frac{\sqrt{n(n-1)\cdots (n-m+1)}}{n!}\sum_{\pi\in{\mathcal S}_n}f_1(k_{\pi(1)})\cdots f_m(k_{\pi(m)}){\rm P}\!_{n-m}\psi(k_{\pi(m+1)},\ldots,k_{k_{\pi(n)}}).
	\end{aligned}
	\label{def:cre}
\end{align}
They satisfy the canonical commutation relations $[a(f),a(g)]=[\ad(f),\ad(g)]=0$ and $[a(f),\ad(g)]=\braket{f,g}$ for $f,g\in L^2(\Omega)$ on a dense domain.
Further, it is well-known that ${\mathscr D}(a(f))={\mathscr D}(\ad(f))$ for any $f\in L^2(\Omega)$.

If $\omega:\Omega\to[0,\infty)$ is a selfadjoint invertible multiplication operator, i.e., satisfies $\omega>0$ a.e., and if $f\in{\mathscr D}(\omega^{-1/2})$ then ${\mathscr D}(a(f))\subset {\mathscr D}({\rm d}\Gamma(\omega))$ and
\begin{align}\label{eq:abound}
	\norm{a(f)\psi} \le \|\omega^{-1/2}f\|\|{\rm d}\Gamma(\omega)\psi\|, \qquad \psi\in{\mathscr D}(a(f)).
\end{align}
The field operator is the selfadjoint operator
\begin{align}
	\ph(f) \coloneqq \overline{a(f)+\ad(f)},
\end{align}
where $\overline{\,\cdot\,}$ here denotes the operator closure.

Finally, we want to introduce exponentials of creation operators. Since these are usually unbounded, we will in fact directly define the operator closure of ${\mathrm e}^{\ad(f)}{\mathrm e}^{-t{\rm d}\Gamma(\omega)}$ for $t>0$, which is bounded by \cref{eq:abound} if $f\in{\mathscr D}(\omega^{-1/2})$ for some multiplication operator $\omega:\Omega\to[0,\infty)$ satisfying $\omega>0$ almost everywhere. We will use the series expansion
\begin{align}\label{def:Ft}
	F_t^\omega(f) \coloneqq \sum_{n=0}^\infty\frac{1}{n!}\ad(f)^n{\mathrm e}^{-t{\rm d}\Gamma(\omega)}.
\end{align}
It is proven in \cite[Appendix 6]{GueneysuMatteMoller.2017} that $F_t^\omega$ defines a bounded operator on ${\mathcal F}(\Omega)$ satisfying
\begin{align}\label{eq:Fbd}
	\norm{F_t^\omega(f)} \le  {\mathrm e}^{4\|f\|_\omega^2}, \qquad f\in{\mathscr D}(\omega^{-1/2}), \qquad \|f\|_\omega^2 \coloneqq \|f\|^2+\|\omega^{-1/2}f\|^2.
\end{align}

\subsection{Nelson Model}\label{subsec:models}
We move to the definition of the translation-invariant Nelson-type model we consider.
In general its full Hamiltonian $\hat H_{\Psi,\omega,v}$ on $L^2({\mathbb R}^d)\otimes {\mathcal F}({\mathbb R}^d)$ is given by
\begin{align}
	\hat H_{\Psi,\omega,v} \coloneqq \Psi(-{\mathrm i}\nabla_x)\otimes\one + \one\otimes{\rm d}\Gamma(\omega) + \ph({\mathrm e}^{-{\mathrm i}\hat p \cdot x}v).
\end{align}
Here, $\Psi:{\mathbb R}^d\to[0,\infty)$ denotes the dispersion relation of the particle, $\omega:{\mathbb R}^d\to [0,\infty)$ with $\omega>0$ a.e. is the dispersion relation of the field bosons and 
$v\in {\mathscr D}(\omega^{-1/2})\cap {\mathscr D}(\omega)$ is the coupling function, sometimes called form factor in the literature. Further, $\hat p$ denotes the momentum operator on the one boson space $L^2({\mathbb R}^d)$, i.e., $\hat p f(k) = k f(k)$ and $x$ is the position operator on the particle Hilbert space.

The operator $\hat H_{\Psi,\omega,v}$ is self-adjoint on ${\mathscr D}(\Psi(-{\mathrm i}\nabla_x)\otimes\one)
\cap {\mathscr D}(\one\otimes{\rm d}\Gamma(\omega))$ and bounded from below, by the Kato--Rellich theorem and the relative bound \cref{eq:abound}. 
Further, it is translation-invariant, i.e., it commutes with the total momentum operator 
\begin{align}
	P_{\rm tot}=-\ri \nabla_x\otimes \one+\one\otimes{\rm d}\Gamma(\hat p).
\end{align}
Hence, we can decompose $\hat H_{\Psi,\omega,v}$ with respect to the joint spectrum of $P_{\rm tot}$ as 
\begin{align}\label{PP}
	\hat H_{\Psi,\omega,v}\cong \int_{{\mathbb R}^{d}}^\oplus H_{\Psi,\omega,v}(P) {\rm d}P.
\end{align}
Here $\cong$ denotes unitary equivalence on $L^2({\mathbb R}^d)\otimes {\mathcal F}({\mathbb R}^d)$, which is
explicitly implemented by the Lee--Low--Pines \cite{LeeLowPines.1953} operator $U$ given by
\begin{align}
	\label{U}
	U F(P)=\frac{1}{(2\pi)^{d/2}}\int_{{\mathbb R}^{d}} e^{-\ri x\cdot(P-{\rm d}\Gamma(\hat p))}F(x)\rd x
\end{align}
for $F\in L^2({\mathbb R}^d)\otimes {\mathcal F}({\mathbb R}^d)=L^2({\mathbb R}^d,{\mathcal F}({\mathbb R}^d))$ and 
\begin{align}
	\int_{{\mathbb R}^d} \Braket{U F(P), e^{-TH_{\Psi,\omega,v}(P)} U G(P)}_{\mathcal F} \rd P=
	\Braket{F, e^{-T\hat H_{\Psi,\omega,v}}G}_{{\mathcal H}}.
\end{align}
Therein, the translation-invariant Nelson Hamiltonian $H_{\Psi,\omega,v}(P)$ with total momentum $P\in{\mathbb R}^{d}$ is the selfadjoint operator acting on ${\mathcal F}({\mathbb R}^d)$ with domain ${\mathscr D}(\Psi(P-{\rm d}\Gamma(\hat p)))\cap {\mathscr D}({\rm d}\Gamma(\omega))$ satisfying
\begin{align}\label{def:Nelson}
	H_{\Psi,\omega,v}(P)=\Psi(P-{\rm d}\Gamma(\hat p))+{\rm d}\Gamma(\omega) + \ph(v),\quad P\in{\mathbb R}^{d}.
\end{align} 
This is the operator, which we investigate in this article.
More concretely, for the purpose of this \lcnamecref{sec:results}, we consider the following two cases:
\begin{itemize}
	\item[(1)]The {\em non-relativistic Nelson model} in $d=3$ dimensions is given by the choice 
\[\Psi(p) = \Psi_\nr(p) \coloneqq \frac12\abs p^2,
\quad
\omega(k)=\sqrt{|k|^2+m^2},
\quad
v(k)=v_\Lambda(k)=\lambda\chr_{|k|<\Lambda}\omega(k)^{-1/2},\]
where $m\ge0$ is the boson mass, $\lambda\in\mathbb R$ is the coupling constant and $\Lambda\in(0,\infty)$ is an ultraviolet cutoff. In this case, we denote $H_{\Psi,\omega,v_\Lambda}(P)$ by $H_{\nr,\Lambda}(P)$.
	\item[(2)] The {\em semi-relativistic Nelson model} in $d=2$ dimensions is given by the choice 
	\[\Psi(p)= \Psi_\sr(p) \coloneqq \sqrt{\abs p^2+M^2}-M,\quad \omega(k)=\sqrt{|k|^2+m^2}, \quad v(k)=v_\Lambda(k)=\lambda\chr_{|k|<\Lambda}\omega(k)^{-1/2},\]
	where 
	$M\ge 0$ is the particle mass, $m>0$ is the boson mass, and $\lambda,\Lambda$ are as before. In this case, we denote $H_{\Psi,\omega,v_\Lambda}(P)$ by $H_{\sr,\Lambda}(P)$. We emphasize that the bosons can not be massless in this case, cf. \cite{HinrichsMatte.2023} for a discussion of this fact.
\end{itemize}
To remove the ultraviolet cutoff $\Lambda$, we introduce the renormalization energy
\begin{align}
	E_\Lambda \coloneqq -\|(\Psi+\omega)^{-1/2}v_\Lambda\|^2.
\end{align}
Note that $E_\Lambda$ does not depend on the total momentum $P\in{\mathbb R}^d$.

The ultraviolet renormalized Nelson Hamiltonian now is the unique selfadjoint lower-semibounded Hamiltonian $H_\#(P)$, $\#\in\{\nr,\sr\}$ satisfying
\begin{align}\label{eq:UVren}
	{\mathrm e}^{-tH_\#(P)} = \lim_{\Lambda\to\infty} {\mathrm e}^{-t(H_{\#,\Lambda}(P)-E_\Lambda)}, \qquad t>0,
\end{align}
where the limit is in operator norm. Proofs for the existence of this operator can be found in \cite{Cannon.1971,Hiroshima.2015,MatteMoller.2018,Lampart.2020,DamHinrichs.2021} in the case $\#=\nr$ and \cite{Sloan.1974,HinrichsMatte.2023} in the case $\#=\sr$. We remark that the norm convergence of the full operators follows from this fact, see the discussion in \cite{HinrichsMatte.2023}, and has separately been proven in the further references given in the introduction.
We also remark that a similar result is not expected to hold for the semi-relativistic model in $d=3$ dimensions, cf. \cite{DeckertPizzo.2014}, also see \cite{Gross.1973} for a treatment of the full operator in this case with a different renormalization method.

\subsection{The Fr\"ohlich Cone}\label{subsec:cone}

Our notion of positivity is induced by the Fr\"ohlich cone
\begin{align}
	\FP(\Omega) = \{\psi\in{\mathcal F}(\Omega)|{\rm P}\!_n\psi \in L^2_+(\Omega^n)\ \text{for all $n\in{\mathbb N}_0$}\},
\end{align}
where $L^2_+(M)\coloneqq \{f\in L^2(M)|f>0\ \text{almost everywhere}\}$.
Its dual and closure coincide, whence the convex cone
\begin{align}
	\FBP(\Omega) = \{\psi\in{\mathcal F}(\Omega)|{\rm P}\!_n\psi \in \bar L^2_+(\Omega^n)\ \text{for all $n\in{\mathbb N}_0$}\}
\end{align}
with $\bar L^2_+(M) = \{f\in L^2(M)|f\ge 0\ \text{almost everywhere}\}$
is self-dual.

Recalling \cref{eq:Fockres,Qexplicit}, we will utilize the rather straightforward observations
\begin{align}\label{eq:posdecomp}
	\psi\in\FP(\Omega) \iff \forall \Theta\subset\Omega: {\rm Q}^\Omega_\Theta \psi \in \FP(\Theta) \Rightarrow  \forall \Theta \subset \Omega: (1-{\rm Q}^\Omega_\Theta)\psi\in\FBP(\Omega).
\end{align}
 We will call a bounded operator $B\in{\mathcal B}({\mathcal F}(\Omega))$ {\em positivity preserving} if $B\FBP(\Omega)\subset \FBP(\Omega)$ and {\em positivity improving} if $B(\FBP(\Omega)\setminus\{0\})\subset \FP(\Omega)$. 

\subsection{Ergodicity of Nelson Semigroups}\label{subsec:erg}
We can now state our main result, which will be proven in the end of \cref{subsec:Nelson}.
\begin{thm}
	\label{thm:Nelsonpos}
	For $\#\in\{\nr,\sr\}$, $P\in{\mathbb R}^d$ and $t>0$, if we have negative coupling constant $\lambda<0$, then the operator ${\mathrm e}^{-tH_{\#}(P)}$ is positivity improving.
\end{thm}
\begin{rem}
	Let us compare this result to previous results in the literature. In the ultraviolet-regularized and non-relativistic case, i.e., for the operator $H_{\#,\Lambda}(P)$ with $\Lambda\in(0,\infty)$ as introduced above, the result was proven in \cite{Frohlich.1973,Frohlich.1974} for the non- and in \cite{Moller.2005} for both the non- and the semi-relativistic case. As far as the renormalized operator goes, the non-relativistic case $\#=\nr$ was recently treated in \cite{Miyao.2018,Miyao.2019,Lampart.2020} by different methods. In the semi-relativistic case, the result is new to the authors knowledge.
\end{rem}
\begin{rem}
	We also note that there is a second notion of positivity in Fock space, induced by the Schr\"odinger or ${\mathcal Q}$-space representation, see for example \cite{HiroshimaLorinczi.2020}. In this setting, already the free Hamiltonian ($v=0$) is positivity improving, but only in the case $P=0$. For treatments of the interacting operator at $P=0$, see for example \cite{Gross.1972,Sloan.1974,Hiroshima.2007}.
\end{rem}
\subsection{Feynman-Kac formulas}
\label{subsec:FK}
Essential to our argument is the functional integral representation of the semigroup ${\mathrm e}^{-tH_\#(P)}$ and its regularized counterpart ${\mathrm e}^{-tH_{\#,\Lambda}(P)}$ given by a Feynman--Kac formula, which we now want to introduce. Their concrete form is the motivation for our study of more general functional integral expressions in \cref{sec:positivity}.

Let $({\mathcal X}, {\mathcal B}, {\mathbb P})$ be a probability space and 
denote by ${\mathbb E}$ the expectation with respect to ${\mathbb P}$. 
We will assume that $X=(X_t)_{t\geq0}$ is a L\'evy process on this probability space. 
Then
there exists a unique characteristic exponent $\Psi$, also called L\'evy symbol of $X$, such that 
 \begin{align}
\label{L} 	{\mathbb E}\Big[{\mathrm e}^{{\mathrm i}k\cdot X_t}\Big] = {\mathrm e}^{-t\Psi(k)}.
 \end{align}
It is known as the L\'evy--Khintchine formula 
that $\Psi$ is of the form 
\begin{align}\label{LL}
\Psi(k)= ib\cdot k-\frac{1}{2}k\cdot A k+\int_{{\mathbb R}^d\setminus\{0\}}
(e^{ik\cdot y}-1-ik\cdot y\one_{\{|y|<1\}})\nu(\rd y),
\end{align}
where $b\in{\mathbb R}^d$, $A$ is a $d\times d$-nonnegative definite symmetric matrix, 
and $\nu$ a L\'evy measure.  
Conversely if a triplet $(b, A, \nu)$ of this type is given, then  
there exists unique L\'evy process $X$ such that 
\cref{L,LL} are satisfied,
see for example \cite[Section~2.4.1]{LorincziHiroshimaBetz.2020} or \cite[Section~1.2]{Applebaum.2009} for details.

For $\#\in\{\nr,\sr\}$, let $X^{\#}=(X^{\#}_t)_{t\geq0}$ be the L\'evy process  on the probability space with L\'evy symbol  $\Psi_\#$ as defined above, i.e.,
 \begin{align*}
 	{\mathbb E}\Big[{\mathrm e}^{{\mathrm i}k\cdot X^\nr_t}\Big] = {\mathrm e}^{-\frac12t|k|^2}
 	\quad\mbox{and}\quad
 	{\mathbb E}\Big[{\mathrm e}^{{\mathrm i}k\cdot X^\sr_t}\Big] = {\mathrm e}^{-t(\sqrt{|k|^2+M^2}-M)} .
 \end{align*}
Explicitly, in the case $d=3$, $\#=\nr$ the process $X^\#$ is a three-dimensional Brownian motion, i.e., the L\'evy process with triplet $(0,\one_{\mathbb C^3},0)$.
Further, in the case $d=2$, $\#=\sr$ the process $X^\#$  is a two-dimensional inverse Gaussian process \cite[Example~1.3.21]{Applebaum.2009} or \cite[Examples~2.17,~2,18]{LorincziHiroshimaBetz.2020} with the particle mass $M$, i.e., the L\'evy process with triple $(0,0,\nu_M)$ where the density of the L\'evy measure $\nu_M$ can for example be found in \cite[Section~1.2.6]{Applebaum.2009} or \cite[(2.8)]{HinrichsMatte.2022}.

The Feynman--Kac formulas that we apply now take the form
	\begin{align}\label{eq:FK}
		{\mathrm e}^{-tH_\#(P)} = {\mathbb E}\Big[{\mathrm e}^{u_t^\#}F^\omega_{t/2}(U_t^{\#,-})F^\omega_{t/2}(U_t^{\#,+})^*{\mathrm e}^{{\mathrm i}(P-{\rm d}\Gamma(\hat p))\cdot X^\#_t}\Big],
	\end{align}
	and are proven in \cite[Theorem~7.6]{MatteMoller.2018} for $\#=\nr$ and \cite[Theorem~7.4]{HinrichsMatte.2023} for $\#=\sr$.
	Therein, 
	$u^\#=(u_t^\#)_{t\geq0}$ is a ${\mathbb C}$-valued stochastic process satisfying
	\begin{align}\label{eq:uconvexample} {\mathbb E}\Big[ \sup_{s\in[0,t]}
	\left|{{\mathrm e}^{u^\#_{\Lambda,s}-s E_\Lambda}-{\mathrm e}^{u^\#_s}}\right|^p\Big]\xrightarrow{\Lambda\to\infty}0,\quad t>0,p\ge1, \end{align}
	where
	\begin{align}\label{eq:uUV} u^\#_{\Lambda,T} = \int_0^T\int_0^T \braket{{\mathrm e}^{-{\mathrm i}\hat p\cdot X^\#_t}v_\Lambda|{\mathrm e}^{-{\mathrm i}\hat p\cdot X^\#_s}{\mathrm e}^{-|t-s|\omega}v_\Lambda}{\rm d}s{\rm d}t, \end{align}
	see \cite[Theorem~4.9]{MatteMoller.2018} and \cite[Theorem~6.8]{HinrichsMatte.2023},
	and $U^{\#,\pm}=(U^{\#,\pm}_t)_{t\geq0}$ are the ${\mathscr D}(\omega^{-1/2})$-valued stochastic processes
	 satisfying
	\begin{align}\label{eq:Uconvexample} {\mathbb E}\Big[\sup_{s\in[0,t]}\|U_{\Lambda,s}^{\#,\pm}-U_s^{\#,\pm}\|^2_\omega\Big] \xrightarrow{\Lambda\to\infty} 0 , \end{align}
	where
	\begin{align}\label{eq:UUV}
		U_{\Lambda,T}^{\#,-} = -\int_0^T{\mathrm e}^{-t\omega-{\mathrm i}\hat p \cdot X^\#_t}v_\Lambda{\rm d}t
		\qquad\text{and}\qquad
		U_{\Lambda,T}^{\#,+} = -\int_0^T {\mathrm e}^{-(T-t)\omega -{\mathrm i}\hat p\cdot X^\#_t}v_\Lambda{\rm d}t,
	 \end{align}
	 see \cite[Lemma~3.6]{MatteMoller.2018} and \cite[Lemma~B.1]{HinrichsMatte.2022}.
	The existence of this limit and explicit expressions for $u^\#$ and $U^{\#,\pm}$ are also presented in the referenced articles.

	We also stress that, in the ultraviolet regular case $\Lambda<\infty$, we have the Feynman--Kac formula
	\begin{align}\label{eq:FKUVreg}
		{\mathrm e}^{-tH_{\#,\Lambda}(P)} = {\mathbb E}\Big[{\mathrm e}^{u^\#_{\Lambda,t}}F^\omega_{t/2}(U_{\Lambda,t}^{\#,-})F^\omega_{t/2}(U_{\Lambda,t}^{\#,+})^*{\mathrm e}^{{\mathrm i}(P-{\rm d}\Gamma(\hat p))\cdot X^\#_t}\Big].
	\end{align}
	This is easily recognized to be a special case of \cref{prop:evolution} below, for which we sketch a proof in \cref{appendix}.
	In fact, the convergence statements \cref{eq:UVren} can then be deduced from the convergence in \cref{eq:uconvexample,eq:Uconvexample} and the Feynman--Kac formulas \cref{eq:FK,eq:FKUVreg}, the main argument which was used in the articles \cite{MatteMoller.2018,HinrichsMatte.2022}.

\section{Positivity of Fock Space Functional Integrals}
\label{sec:positivity}
In this \lcnamecref{sec:positivity}, we prove positivity of functional integral expressions of the form \cref{eq:FK}.
We present three theorems:
\begin{itemize}
	\item[(1)] \cref{thm:regular} presented in \cref{subsec:Uvreg} treats Nelson-type models in presence of an ultraviolet cutoff. Whereas ergodicitiy in this case is well-known, the novelty of our approach is to demonstrate that it can be inferred solely using the functional integration representation of the model. Further, \cref{thm:regular} is a major ingredient to our proofs of the subsequent results.
	\item[(2)] \cref{thm:posren} presented in \cref{subsec:formren} extends \cref{thm:regular} to the case of functional integrals only given by a limiting procedure. Whereas this result does not suffice to cover the Nelson model, we can apply it to reprove ergodicity of the Fr\"ohlich Hamiltonian in \cref{fr}.
	\item[(3)] \cref{thm:posselfen} presented in \cref{subsec:Nelson} is the main result leading to our proof of \cref{thm:Nelsonpos}, which is given in the end of that section. Compared to the previous results, we require a 
	Trotter-type product formula for the functional integrals, which we can only prove using an evolution equation presented in \cref{prop:evolution}. It would be an interesting question to derive such a Trotter product formula (explicitly stated in \cref{cor:trotter}) directly from the path integral representation.
\end{itemize}
Let us introduce some assumptions and notation used throughout this \lcnamecref{sec:positivity}.

As in \cref{subsec:FK}, we fix a probability space $({\mathcal X}, {\mathcal B}, {\mathbb P})$ and 
denote by ${\mathbb E}$ the expectation with respect to ${\mathbb P}$. 
Further, on this probability space, we assume that $X=(X_t)_{t\ge0}$ is a L\'evy process
with c\`adl\`ag paths and a real-valued characteristic function, i.e., that the characteristic exponent $\Psi:{\mathbb R}^d\to {\mathbb C}$ given by
\begin{align}\label{eq:Levysymbol}
	{\mathbb E}\Big[{\mathrm e}^{{\mathrm i}k\cdot X_t}\Big] = {\mathrm e}^{-t\Psi(k)}
\end{align}
is real-valued and hence also non-negative by the L\'evy--Khintchine formula \cref{LL}, since it directly implies $\operatorname{Re} \Psi(k)\ge 0$, $k\in\mathbb R^d$.

A key ingredient to our proof will be that the independent and stationary increments of a L\'evy process imply
\begin{align}\label{eq:charfunct2}
	{\mathbb E}\Big[\prod_{j=1}^n{\mathrm e}^{{\mathrm i}k_j\cdot X_{t_j}}\Big] = \prod_{j=1}^n {\mathrm e}^{-(t_j-t_{j-1})\Psi(\sum_{i=j}^n k_i)}, \quad 0=t_0<t_1<t_2<\cdots<t_n.
\end{align}
We also assume $\omega:{\mathbb R}^d\to[0,\infty)$ to be an invertible multiplication operator on $L^2(\Omega)$ for any measurable $\Omega\subset {\mathbb R}^d$, i.e., $\omega$ is measurable and $\omega>0$ almost everywhere.
We denote the domain ${\mathscr D}(\omega^{-1/2})$ in $L^2(\Omega)$ by ${\mathfrak D}_\omega(\Omega)$ and equip it with the graph norm of $\omega^{-1/2}$, i.e., $\|f\|_\omega^2  \coloneqq \|f\|^2 + \|\omega^{-1/2}f\|^2$ for $f\in{\mathfrak D}_\omega(\Omega)$.

Given any measurable $\Omega\subset {\mathbb R}^d$, we abbreviate the Bochner spaces
\begin{align}
\begin{array} {l}\displaystyle
{\mathfrak l}_{2}(\Omega)\coloneqq L^2_{\mathsf{loc}}([0,\infty)^2;L^2(\Omega)),\\
\ \\
{\mathfrak l}_{1,\omega}(\Omega) \coloneqq L^1_{\mathsf{loc}}([0,\infty);{\mathfrak D}_\omega(\Omega)).
\end{array}\end{align}
As usual, we will write $f(t)$, $t\ge 0$ for a.e. well-defined representatives of equivalence classes in these spaces.
Then given $S,T>0$ and $\alpha,\beta\in {\mathfrak l}_2(\Omega)$, $v\in {\mathfrak l}_{1,\omega}(\Omega)$, we can define
the associated random variables
\begin{align}\label{def:pathintegral} 
\begin{array} {l}\displaystyle
u_{\alpha,\beta}(S,T) \coloneqq \int_{S}^{T}\int_{S}^T \braket{{\mathrm e}^{-i\hat p\cdot X_t}\alpha(t,s)|{\mathrm e}^{-i\hat p\cdot X_s}\beta(t,s)}{\rm d}s{\rm d}t,\\
 \\
\displaystyle U_{v}(S,T) \coloneqq \int_{S}^{T} {\mathrm e}^{-{\mathrm i}\hat p\cdot X_t}v(t){\rm d}t  
 \end{array}
 \end{align}
as Lebesgue integral and  ${\mathfrak D}_\omega(\Omega)$-valued Bochner--Lebesgue integrals, respectively. 

\subsection{The Ultraviolet Regular Case}\label{subsec:Uvreg}

Our first result combined with the Feynman--Kac formula \cref{eq:FKUVreg} proves that the semigroup generated by the regularized Nelson model at negative coupling $\lambda<0$ is positivity improving for any total momentum $P$.
\begin{thm}\label{thm:regular}
	Let $N\in{\mathbb N}_0$, let $\Omega\subset {\mathbb R}^d$ measurable, let $v^\pm\in{\mathfrak l}_{1,\omega}(\Omega)$ with $v^\pm(t)\in \bar L^2_+(\Omega)$, $t\ge0$ and let $\alpha^\pm,\beta^\pm\in {\mathfrak l}_2(\Omega)$ satisfy $\alpha^\pm(s,t),\beta^\pm(s,t)\in \bar L^2_+(\Omega)$, $s,t\ge0$. Then the ${\mathcal B}({\mathcal F}(\Omega))$-valued Bochner--Lebesgue integral
	\begin{align}\label{eq:pospres}
		S(P)\coloneqq{\mathbb E}\Big[{\mathrm e}^{u_{\alpha^-,\alpha^+}(\sigma_1,\tau_1)}
		u_{\beta^-,\beta^+}(\sigma_2,\tau_2)^NF^\omega_\gamma(U_{v^-}(\sigma_3,\tau_3))F^\omega_\rho(U_{v^+}(\sigma_4,\tau_4))^*{\mathrm e}^{{\mathrm i}(P-{\rm d}\Gamma(\hat p))\cdot (X_{T_2}-X_{T_1})} \Big]
	\end{align}
	is well-defined and positivity preserving for all $P\in{\mathbb R}^d$ and $\gamma,\rho,\sigma_1,\tau_1,\ldots,\sigma_4,\tau_4,T_1,T_2>0$. 
	
	Further, if $\alpha^\pm(s_1,t_1),v^-(t_3),v^+(t_4) \in L^2_+(\Omega)$, $s_1,t_1\in[\sigma_1,\tau_1],t_3\in[\sigma_3,\tau_3],t_4\in[\sigma_4,\tau_4]$ and either $N=0$ or $\beta^\pm(t,s)\in L^2_+(\Omega)$, $t,s\in[\sigma_2,\tau_2]$, then $S(P)$ is positivity improving for all $P\in{\mathbb R}^d$.
\end{thm}
\begin{proof}
	Throughout this proof, we  abbreviate $u_\alpha \coloneqq u_{\alpha^-,\alpha^+}(\sigma_1,\tau_1)$, $u_\beta \coloneqq u_{\beta^-,\beta^+}(\sigma_2,\tau_2)$, $U^- \coloneqq U_{v^-}(\sigma_3,\tau_3)$ and $U^+ \coloneqq U_{v^+}(\sigma_4,\tau_4)$.
	The proof is divided into three steps.
	
	\smallskip\noindent{\em Step 1.} We first argue for the well-definedness of $S(P)$.
	
	It will follow from Pettis' theorem combined with the standard criterion for Bochner integrability.
	The measurability of the argument of the expectation in \cref{eq:pospres} (i.e., weak measurability and 
	almost surely separable valuedness) hereby follows from the separability of $L^2(\Omega)$,
	continuity of the maps $(0,\infty)\times {\mathfrak D}_\omega(\Omega)\ni(s,h)\mapsto F^\omega_s(h)\in{\mathcal B}({\mathcal F}(\Omega))$ and $\Omega\ni x\mapsto {\mathrm e}^{{\mathrm i}(P-{\rm d}\Gamma(\hat p))\cdot x}$ as well as our definition \cref{def:pathintegral}.
	Further, by \cref{eq:Fbd} and our integrability assumptions, we easily see that the 
	${\mathbb P}$-integrand in \cref{eq:pospres} is uniformly bounded in the paths of $X$ and thus the claim of this step.
	
	\smallskip\noindent{\em Step 2.}
	Let us now prove that $S(P)$ is positivity preserving, by proving 
	that every summand of
	\begin{align*}
		{\rm P}\!_nS(P)\psi = \sum_{m=0}^\infty {\rm P}\!_nS(P){\rm P}\!_m\psi
	\end{align*}
	 is an element of $\bar L^2_+(\Omega^n)$. We will now assume $m,n\in{\mathbb N}_0$ to be fixed.
	
	To this end, we first note that by the definitions \cref{def:ann,def:cre,def:Ft}, for any $r,s>0$ and $f,g\in{\mathfrak D}_\omega(\Omega)$
	\begin{align*}
		&{\rm P}\!_nF^\omega_r(f)F^\omega_s(g)^*{\mathrm e}^{{\mathrm i}(P-{\rm d}\Gamma(\hat p))\cdot (X_{T_2}-X_{T_1})}{\rm P}\!_m\psi  \\& \qquad= \sum_{i=0}^{m\wedge n}\frac{1}{(n-i)!(m-i)!}\ad(f)^{n-i}{\mathrm e}^{-(r+s)\dGn i(\omega)}a(g)^{m-i}{\mathrm e}^{{\mathrm i}(P-\dGn m(\hat p))\cdot (X_{T_2}-X_{T_1})}{\rm P}\!_m\psi,
	\end{align*}
	where $a\wedge b \coloneqq \min(a,b)$ as usual.
	Further employing the series expansion of the exponential ${\mathrm e}^{u_\alpha}$ 
	and Fubini's theorem, we find
	\begin{align*}
		&{\rm P}\!_nS(P){\rm P}\!_m\psi  \\&=  \sum_{i=0}^{m\wedge n}\frac{1}{(n-i)!(m-i)!}\sum_{\ell=0}^\infty\frac{1}{\ell!}{\mathbb E}\Big[u_\alpha^\ell u_\beta^N\ad(U^-)^{n-i}{\mathrm e}^{-{(r+s)}\dGn i(\omega)}a(U^+)^{m-i}{\mathrm e}^{{\mathrm i}(P-\dGn m(\hat p))\cdot (X_{T_2}-X_{T_1})}{\rm P}\!_m\psi\Big].
	\end{align*}
	Inserting the definition of the creation and annihilation operators \cref{def:ann,def:cre} and once more applying Fubini's theorem, we find for $k\in\Omega^n$
	\begin{align}\label{eq:intkernel}
		\begin{aligned}
			&{\rm P}\!_nS(P){\rm P}\!_m\psi(k) 
			\\&\quad
			=\sum_{\pi\in{\mathcal S}_n} \sum_{i=0}^{m\wedge n} C_{m,n,i} \sum_{\ell=0}^\infty \frac1{\ell!} \int_{{\mathbb S}}\int_{\TT} F^{\pi,i,\ell}_{k,p,t}{\mathbb E}\left[
			{\mathrm e}^{-{\mathrm i}R^{\pi,\ell}_{k,p,t}}\right] 
			({\rm P}\!_m\psi)(k_{\pi(1)},\ldots,k_{\pi(i)},p_1,\ldots,p_{m-i}){\rm d}t{\rm d}p,
		\end{aligned}
	\end{align}
	where the $p$- and $t$-integrals run over
	\begin{align*}
		{\mathbb S} = {\Omega^{\ell+N+m-i}} \qquad \mbox{and}\qquad \TT= [\sigma_1,\tau_1]^{2\ell} \times [\sigma_2,\tau_2]^{2N} \times [\sigma_3,\tau_3]^{n-i}\times [\sigma_4,\tau_4]^{m-i},
	\end{align*}
	respectively,
	the constants are explicitly given by
	\begin{align}
		C_{m,n,i} \coloneqq \begin{cases}\frac{\sqrt{n(n-1)\cdots(i+1)}}{(n-i)!(m-i)!n!}, & \ i<n,\\\frac1{(m-i)!}, & \ i=n,	\end{cases}
	\end{align}
	 and we have defined the integral kernels
	\begin{align}\label{eq:Fdef}
		&\begin{aligned}
			F^{\pi,i,\ell}_{k,p,t} \coloneqq &\prod_{j=1}^{\ell} \alpha^+(t_{2j-1},t_{2j},p_{j})\alpha^-(t_{2j-1},t_{2j},p_{j})\\&\times
				\prod_{j=\ell+1}^{\ell+N} \beta^+(t_{2j-1},t_{2j},p_{m+n+j})\beta^-(t_{2j-1},t_{2j},p_{m+n+j})\\
				&\times \prod_{h=1}^{n-i}v^+(t_{2(\ell+N)+h},k_{\pi(i+h)})\prod_{h=1}^{m-i}v^-(t_{2(\ell+N)+n-i+h},p_{\ell+N+h}),
		\end{aligned}
	\end{align}
	as well as the random variable
	\begin{align}
		\begin{aligned}
			R^{\pi,i,\ell}_{k,p,t} \coloneqq  &\sum_{j=1}^{\ell+N}p_j\cdot(X_{t_{2j}}-X_{t_{2j-1}}) + \sum_{h=1}^{m-i} p_{\ell+N+h} \cdot X_{t_2(\ell+N)+n-i+h}
			\\& + \sum_{h=1}^{n-i}k_{\pi(i+h)}\cdot X_{t_2(\ell+N)+h}  + \left(\sum_{j=1}^{m}k_j - P\right)\cdot (X_{T_2}-X_{T_1}).
		\end{aligned}
	\end{align}
	By the assumptions on $\alpha$, $\beta$ and $v$, the kernel $F^{\pi,i,\ell}_{k,p,t}\ge 0$ for almost all 
	$(k,p,t)\in \Omega^n\times{\mathbb S}\times{\mathbb T}$. 
	Further, ${\mathbb E}[{\mathrm e}^{-{\mathrm i}R^{\pi,i,\ell}_{k,p,t}}]>0$ for all $(k,p,t)\in \Omega^n\times{\mathbb S}\times{\mathbb T}$ in view of \cref{eq:charfunct2}, which proves the statement.
	
	\smallskip\noindent{\em Step 3.}
	It remains to prove that $S(P)$ is positivity improving, if we use the stronger assumptions from the theorem.
	These immediately imply that the integral kernel given in \cref{eq:Fdef} satisfies
	$F^{\pi,i,\ell}_{k,p,t}>0$ for all $m,n,\pi,i,\ell,t$ and almost all $(k,p,t)\in \Omega^n\times{\mathbb S}\times{\mathbb T}$.
	If for any fixed $n\in{\mathbb N}_0$ there exists $m\in{\mathbb N}$ such that one of the summands ${\rm P}\!_nS(P){\rm P}\!_m\psi$ is an element of $L^2_+(\Omega^m)$, the proof is complete.
	This is the case, if the right hand side of \cref{eq:intkernel} is strictly positive for (almost) all $k\in\Omega^n$, which is now easily observed to hold if we choose $m=n$ and $i=0$.
\end{proof}
\subsection{Mild Ultraviolet Divergence Cases}\label{subsec:formren}
In this section, we prove an intermediate result, which does not cover self-energy renormalizations of the type \cref{eq:UVren} discussed in \cref{subsec:Nelson} below, 
but generalizes and simplifies the proof of ergodicity, e.g., for the Fr\"ohlich polaron model. 

The three-dimensional translation-invariant Fr\"ohlich polaron \cite{Frohlich.1954} at total momentum 
$P\in{\mathbb R}^3$ with ultraviolet cutoff $\Lambda>0$ is given by the Hamiltonian 
\begin{align}\label{eq:Fr}
	H_{{\sf F},\Lambda}(P) \coloneqq H_{\Psi,\omega,v_\Lambda}(P)
	\end{align}
with the choices 
\begin{align}\label{eq:Frohlich}\Psi=\Psi_\nr,\quad \omega = 1, \quad v_\Lambda(k)=\lambda\chr_{|k|<\Lambda}|k|^{-1}.\end{align}
It is well-known that the semigroup 
${\mathrm e}^{-tH_{{\sf F},\Lambda}(P)}$ converges to ${\mathrm e}^{-tH_{{\sf F},\infty}(P)}$  in norm as $\Lambda\to\infty$, i.e.,
\begin{align}\label{eq:convfr}
\lim_{\Lambda\to\infty}{\mathrm e}^{-tH_{{\sf F},\Lambda}(P)}=
{\mathrm e}^{-tH_{{\sf F},\infty}(P)},
\end{align}
where $H_{\sf F,\infty}(P)$ is still given by the expression \cref{def:Nelson}, but now in the sense of quadratic forms,
see for example \cite{LiebYamazaki.1958}.
	In \cite{HinrichsMatte.2024}, Matte and the first author discuss how the limiting operator can explicitly be represented by a Feynman--Kac formula of the type \cref{eq:FK}. 
	In fact, in that article, the Fr\"ohlich polaron is discussed in presence of an external potential. However, the translation-invariant case can directly be inferred from \cref{prop:evolution} below and the results therein.

\begin{thm}\label{thm:posren}
	For $k=1,\ldots,4$, let $t_k\ge s_k\ge 0$, assume $\Omega_1\subset\Omega_2\subset \cdots\subset{\mathbb R}^d$ is an increasing sequence of measurable subsets  and set $\Omega= \bigcup_{n\in{\mathbb N}}\Omega_n$. Further, let $(v_n^\pm)_{n\in{\mathbb N}}\subset {\mathfrak l}_{1,\omega}(\Omega)$ and $(\alpha_n^\pm)_{n\in{\mathbb N}}\subset {\mathfrak l}_2(\Omega)$ and, for all $m\ge n\in{\mathbb N}$, $t,s\ge 0$, assume that 
$\alpha^\pm_n(t,s)\lceil_{\Omega_n} \in L^2_+(\Omega_n)$ and 
	\begin{align}\label{eq:simplemon}
		\begin{aligned}
			\lim_{\substack{k\to\infty\\L^2(\Omega_n)}} v^\pm_k(t)\lceil_{\Omega_n}\in L^2_+(\Omega_n),\qquad 
\qquad 			\alpha^\pm_m(t,s)\lceil_{\Omega_n}\ge \alpha^\pm_n(t,s)|_{\Omega_n}. 
		\end{aligned}
	\end{align} 
	Employing the definition \cref{def:pathintegral}, we finally assume there exists a ${\mathbb C}$-valued random variable $u_\infty$,  ${\mathfrak D}_\omega(\Omega)$-valued random variables $U^\pm_{\infty}$ and $p,q,r,r'\ge 1$ with $\frac1p+\frac2q=1$ and $\frac1{r}+\frac1{r'}=1$ such that
	\begin{align}\label{eq:convassum}
		\begin{aligned}
			& {\mathbb E}\left[|{\mathrm e}^{pu_\infty}|\right]<\infty,\\
			& \sup_{n\in{\mathbb N}}{\mathbb E}\left[{\mathrm e}^{4qr(\|U_{v_n^\pm}(s_j,t_j)\|_\omega\vee\|U^\pm_\infty\|_\omega)^2}\right]<\infty,\quad j=2,3,\\
			& \lim_{n\to\infty}{\mathbb E}\left[|e^{{u_{\alpha_n^+,\alpha_n^-}(s_1,t_1)}}-e^{u_\infty}|^p\right]=0,\\
			& \lim_{n\to\infty}{\mathbb E}\left[\|U_{v_n^\pm}(s_j,t_j)-U^\pm_\infty\|_\omega^{qr'}\right]=0,\quad j=2,3.
		\end{aligned}
	\end{align}
	Then the ${\mathcal B}({\mathcal F}(\Omega))$-valued Bochner--Lebesgue integral
	\begin{align}\label{eq:pospresfrohlich}
		S_\infty(P)\coloneqq{\mathbb E}\Big[ {\mathrm e}^{u_\infty}F^\omega_\gamma(U^-_{\infty})F^\omega_\sigma(U^+_{\infty})^*{\mathrm e}^{{\mathrm i}(P-{\rm d}\Gamma(\hat p))\cdot (X_{t_4}-X_{s_4})} \Big]
	\end{align} 
	is well-defined and positivity improving for any $P\in{\mathbb R}^d$ and $\gamma,\sigma>0$.
\end{thm}
\begin{proof}We again divide the proof into several steps and simplify the notation. We set 
$u_n \coloneqq u_{\alpha_{n}^-,\alpha_n^+}(s_1,t_1)$, 
$U_n^- \coloneqq U_{v_n^-}(s_2,t_2)$ and  
$U_n^+ \coloneqq U_{v_n^+}(s_3,t_3)$.
Set \[
S_n(P)\coloneqq {\mathbb E}\Big[{\mathrm e}^{u_{n}}F^\omega_\gamma(U^-_{n})F^\omega_\sigma(U^+_{n})^*{\mathrm e}^{{\mathrm i}(P-{\rm d}\Gamma(\hat p))\cdot (X_{t_4}-X_{s_4})} \Big].\]  	\smallskip\noindent{\em Step 1.}
	We first prove the well-definedness of $S_\infty(P)$ and the norm convergence
	\begin{align}\label{eq:opconv}
		\|S_n(P)-S_\infty(P)\|\xrightarrow{n\to\infty} 0.	\end{align}
	Immediately from \cref{thm:regular} this yields that $S_\infty(P)$ is positivity preserving, since $S_n(P)$ is for any $n\in{\mathbb N}$.
	Our argument follows the arguments in \cite[Proof of Proposition 5.8]{MatteMoller.2018}.
	
	To see that $S_\infty(P)$ is well-defined, we proceed as in Step 1 of the proof of \cref{thm:regular}. The measurability argument thereby remains unchanged, whereas the uniform boundedness of the integrand in \cref{eq:pospres} can be deduced from the convergence assumptions in \cref{eq:convassum}.
	
	To prove the convergence in operator norm, we observe that calculating the Gateaux derivative of $h\mapsto F^\omega_t(h)$ and using the fundamental theorem of calculus as well as the bound \cref{eq:Fbd} yields
	\begin{align*}
		\norm{F_{\gamma}^\omega(f)-F_\gamma^\omega(g)} \le \|f-g\|_\omega{\mathrm e}^{4(\|f\|_\omega^2\vee\|g\|_\omega^2)},
	\end{align*}
	see for example \cite[Lemma~17.4]{GueneysuMatteMoller.2017} for more details.
	Thus, multiply employing H\"olders inequality, the triangle inequality and the unitarity of ${\mathrm e}^{{\mathrm i}(P-{\rm d}\Gamma(\hat p))\cdot x}$ for any $x\in{\mathbb R}^d$, we find
	\begin{align*}
		\|S_n(P)  - S_\infty(P)\|
		\le\ & {\mathbb E}[|{\mathrm e}^{u_{n}}-{\mathrm e}^{u_\infty}|^p]^{1/p}{\mathbb E}[{\mathrm e}^{4q\|U^-_{n}\|_\omega^2}]^{1/q}{\mathbb E}[{\mathrm e}^{4q\|U^+_{n}\|_\omega^2}]^{1/q}
		\\
		& + {\mathbb E}[|{\mathrm e}^{pu_\infty}|]^{1/p}{\mathbb E}[\|U^-_{n}-U^-_\infty\|_\omega^{qs}]^{1/qs}{\mathbb E}[{\mathrm e}^{4qr(\|U^-_{n}\|\vee \|U^-_\infty\|_\omega)^2}]^{1/qr}{\mathbb E}[{\mathrm e}^{4q\|U^+_{n}\|_\omega^2}]^{1/q}
		\\
		& + {\mathbb E}[|{\mathrm e}^{pu_\infty}|]^{1/p}{\mathbb E}[{\mathrm e}^{4q\|U^-_{n}\|_\omega^2}]^{1/q}{\mathbb E}[\|U^+_{n}-U^+_\infty\|_\omega^{qs}]^{1/qs}{\mathbb E}[{\mathrm e}^{4qr(\|U^+_{n}\|_\omega\vee \|U^+_\infty\|_\omega)^2}]^{1/qr}.
	\end{align*}
	The right hand side converges to zero as $n\to\infty$, by our convergence assumptions \cref{eq:convassum}, which finishes this step.
	
	\smallskip\noindent{\em Step 2.}
	We now prove that for any $n\in{\mathbb N}$ and
	$\psi \in {\rm Q}^\Omega_{\Omega_n} {\mathcal F}$, we have
	\begin{align}\label{eq:pathres}
		{\rm Q}^\Omega_{\Omega_n} S_\infty(P) &= {\mathbb E}\Big[{\mathrm e}^{u_\infty}F_\gamma^\omega(\tilde U_n^-)F_\sigma^\omega(\tilde U_n^+)^*{\mathrm e}^{{\mathrm i}(P-{\rm d}\Gamma(\hat p))\cdot (X_{t_4}-X_{s_4})}\psi \Big],\\
		\nonumber \qquad
		&\tilde U_n^- \coloneqq U_{\tilde v_{n,\infty}}^-(s_2,t_2),\quad 
\tilde U_n^+ \coloneqq U_{\tilde v_{n,\infty}^+}(s_3,t_3),\quad		\tilde v_{n,\infty}^\pm (t)\coloneqq  \lim_{\substack{k\to\infty\\L^2(\Omega_n)}} v_k^\pm(t).
	\end{align}
	To this end, for any ${\Theta}\subset\Omega$, we first observe that by the definitions \cref{eq:Fockres,def:ann,def:cre,def:dG}, we have ${\rm Q}^\Omega_{\Theta}\ad(f) = \ad(f\lceil_{\Theta}){\rm Q}^\Omega_{\Theta}$ and ${\rm Q}^\Omega_{\Theta}{\rm d}\Gamma(\omega) = {\rm d}\Gamma(\omega\lceil_{\Theta}){\rm Q}^\Omega_{\Theta}$.
	 Thus, using an induction, we obtain that 
	\begin{align*}
		{\rm Q}^\Omega_{\Theta} F_t^\omega(f) = F_t^\omega(f\lceil_{\Theta}){\rm Q}^\Omega_{\Theta}.
	\end{align*}
	The statement \cref{eq:pathres} follows from the first assumption in \cref{eq:simplemon} combined with the observation
	that
	${\mathrm e}^{{\mathrm i}(P-{\rm d}\Gamma(\hat p))\cdot (X_{t_4}-X_{s_4})}$ leaves ${\mathcal F}({\Theta})$ invariant in the decomposition \cref{eq:Fockres} and the fact that ${\rm Q}^\Omega_{\Theta}$ is a partial isometry with ${\rm Ran} {\rm Q}^\Omega_{\Theta} = {\mathcal F}({\Theta})$ .
	
	\smallskip\noindent{\em Step 3.} We now prove that $S_\infty(P)$ is positivity improving.
	
	To this end let $\psi\in\FBP(\Omega)$  and pick $n_0$ such that ${\rm P}\!_{\Omega_n}\psi \in \overline{{\mathcal F}_+(\Omega_n)}\setminus\{0\}$ 
	for all $n\ge n_0$.
	In view of \cref{eq:Fockres,eq:posdecomp}, we decompose $\psi = {\rm Q}^\Omega_{\Omega_n}\psi + \psi^\perp_n$ with $\psi^\perp_n\in\FBP(\Omega)$ for any $n\in{\mathbb N}$. Recalling \cref{eq:posdecomp} it suffices to prove ${\rm Q}^\Omega_{\Omega_n}S_\infty(P){\rm Q}^\Omega_{\Omega_n}\psi\in\FP(\Omega_n)$ for all $n\ge n_0$. 
	Using Step 2, we see that
	\begin{align}
		{\rm Q}^\Omega_{\Omega_n}S_\infty(P){\rm Q}^\Omega_{\Omega_n}\psi & = {\mathbb E}\Big[{\mathrm e}^{u_{n}+u_{\infty}-u_{n}}F_\gamma^\omega(\tilde U^-_{n})F^\omega_\sigma(\tilde U^+_{n})^*{\mathrm e}^{{\mathrm i}(P-{\rm d}\Gamma(\hat p))\cdot (X_{t_4}-X_{s_4})}{\rm Q}^\Omega_{\Omega_n}\psi \Big]\nonumber\\
			& = \sum_{N=0}^{\infty}\frac{1}{N!}{\mathbb E}\Big[ (u_\infty-u_n)^N {\mathrm e}^{u_{n}}F_\gamma^\omega(\tilde U^-_{n})F_\sigma^\omega(\tilde U^+_{n})^*{\mathrm e}^{{\mathrm i}(P-{\rm d}\Gamma(\hat p))\cdot (X_{t_4}-X_{s_4})}{\rm Q}^\Omega_{\Omega_n}\psi \Big].
			\label{eq:decompsum}
	\end{align}
	Recalling Step 1 and using that convergence of an exponential moment implies convergence of all moments, we see that
	\begin{align}\label{eq:diffsummand}
		\begin{aligned}
			{\mathbb E}\Big[ (u_\infty-u_{\Omega_n})^N &{\mathrm e}^{u_{n}}F_\gamma^\omega(\tilde U^-_{n})F_\sigma^\omega(\tilde U^+_{n}){\mathrm e}^{{\mathrm i}(P-{\rm d}\Gamma(\hat p))\cdot (X_{t_4}-X_{s_4})}{\rm Q}^\Omega_{\Omega_n}\psi \Big]\\
			& = \lim_{m\to\infty} {\mathbb E}\Big[ (u_{m}-u_{n})^N {\mathrm e}^{u_{n}}F_\gamma^\omega(\tilde U^-_{n})F_\sigma^\omega(\tilde U^+_{n})^*{\mathrm e}^{{\mathrm i}(P-{\rm d}\Gamma(\hat p))\cdot (X_{t_4}-X_{s_4})}{\rm Q}^\Omega_{\Omega_n}\psi \Big]. 
		\end{aligned}
	\end{align}
	Now further observing
	\begin{align*}
		u_{m}-u_{n} = u_{\alpha^-_{m}-\alpha^-_n,\alpha^+_m}(s_1,t_1) + u_{\alpha_n^-,\alpha^+_m-\alpha^+_n}(s_1,t_1),
	\end{align*}
	we can
	apply \cref{thm:regular} which yields that the left hand side of \cref{eq:diffsummand} is an element of $\FBP(\Omega)$.
	Furthermore, \cref{thm:regular} implies that the $N=0$ summand of \cref{eq:decompsum} is an element of $\FP(\Omega_n)$. This proves the statement.
\end{proof}
\begin{rem}
	We emphasize that our assumptions in \cref{thm:posren} do not imply convergence of 
	$v_{n}^\pm$ in ${\mathfrak l}_{1,\omega}(\Omega)$ or 
	$\alpha_{n}^\pm$ in ${\mathfrak l}_2(\Omega)$.
	Thus, we appropriately describe renormalized models with mild divergences by the above result.
\end{rem}
\Cref{thm:posren} yields ergodicity of the semigroup of the Fr\"ohlich polaron.
\begin{cor}\label{fr}
	If $\lambda<0$, then
${\rm e}^{-tH_{{\sf F},\infty}(P)}$ is positivity improving for any $P\in{\mathbb R}^3$.
\end{cor}
\begin{proof}
	For the ultraviolet regularized Fr\"ohlich polaron \cref{eq:Fr},
	we have the Feynman--Kac formula
	\begin{align*}
		{\mathrm e}^{-tH_{{\sf F},\Lambda}(P)}
		= \mathbb E\Big[ {\rm e}^{u_{-v_\Lambda,-v_\Lambda}}F_{t/2}^\omega(U_{-{\rm e}^{-\bullet w}v_\Lambda}(0,t))F_{t/2}^\omega(U_{-{\rm e}^{-(t-\bullet)\omega}v_\Lambda}(0,t)){\rm e}^{{\rm i}(P-{\rm d}\Gamma(\hat p))\cdot X_t}\Big],
	\end{align*}
	where $(X_t)_{t\geq0}=X=X_\nr$ is now three-dimensional Brownian motion and $\omega$ and $v_\Lambda$ are defined as in \cref{eq:Frohlich},
	also see \cref{prop:evolution} for a more general statement.
	The required moment and convergence bounds \cref{eq:convassum} for the Fr\"ohlich polaron are proven in \cite[Theorems~5.1 and 6.1]{HinrichsMatte.2024}, so by Step 1 of above proof the semigroup ${\rm e}^{-tH_{{\sf F},\infty}(P)}$ is given by \cref{eq:pospresfrohlich}.\footnote{In fact, by reinspecting Step 1, we see that this provides an alternative proof for the convergence statement \cref{eq:convfr}.}
	Further, if $\lambda<0$, then the positivity assumptions \cref{eq:simplemon} are also easily verified from the definitions \cref{eq:Frohlich}, so
	 \cref{thm:posren} applies and proves the statement.
\end{proof}
\begin{rem}
	Although this result is fairly standard in the ultraviolet regularized case, the authors could not find a direct proof in the literature for the operator without ultraviolet cutoff. We do remark, however, that it can be inferred from the Perron--Frobenius--Faris theorem \cite[Theorem XIII.44]{ReedSimon.1978} for those total momenta $P\in{\mathbb R}^3$ for which the Hamiltonian has a spectral gap, which is known to hold at least for total momenta $P$ with small absolute value, since $\omega$ is massive, i.e., uniformly bounded from below by a positive constant, see for example \cite{Polzer.2023} for a recent discussion of this fact.
\end{rem}
\subsection{Self-Energy Renormalized Cases}\label{subsec:Nelson}
We now move to the study of models which require a self-energy renormalization. A key argument in our main result \cref{thm:posselfen} is the Trotter product formula \cref{cor:trotter}, which will follow from the next \lcnamecref{prop:evolution} providing an evolution equation for our functional integrals.

 To derive the evolution equation, we require more structures of our processes than in the previous results. Thus and throughout this section, we as usually denote by $C(M_1,M_2)$ the space of continuous functions from the topological space $M_1$ to the topological space $M_2$ and by $C^1(I)$ the once continuously differentiable functions on the interval $I\subset {\mathbb R}$.
\begin{prop}[Evolution Equation]\label{prop:evolution}
	Let $g_\pm \in C([0,\infty),{\mathfrak D}_\omega(\Omega))$ and $f \in C^1([0,\infty))$ and define
	\begin{align} h(t,x) = {\rm d}\Gamma(\omega f'(t)) + a({\mathrm e}^{-{\mathrm i}\hat p\cdot x}g_-(t)) + \ad({\mathrm e}^{-{\mathrm i}\hat p \cdot x}g_+(t)) \end{align}
with $f'(t)=\frac{{\rm d}f(t)}{{\rm d}t}$. Then $h(t,x)$ is selfadjoint on ${\mathscr D}({\rm d}\Gamma(\omega))$ for any $t\ge0$ and any 
	$ x\in{\mathbb R}^d$. Further, let
	\begin{align}\label{eq:evolutionassum}
		\begin{aligned}
			 &v^-(s) = {\mathrm e}^{-f(s)\omega}g_-(s),
			&& v^+_t(s) = {\mathrm e}^{-|f(t)-f(s)|\omega}g_+(s),\\
			& \alpha^-(t,s) = 2^{-1/2}v^-(t),
			&& \alpha^+(t,s) = 2^{-1/2}v^+_t(s).
		\end{aligned}
	  \end{align}
	Then the ${\mathcal B}({\mathcal F}(\Omega))$-valued Bochner--Lebesgue integral
	\begin{equation}
		\label{def:SOmega}
		S^\Omega_{s,t}(P) \coloneqq {\mathbb E}\Big[{\mathrm e}^{u_{\alpha^-,\alpha^+}(s,t)}F_{f(t)/2}^\omega(U_{v^-}(s,t))F_{f(t)/2}^\omega(U_{v_t^+}(s,t))^*{\mathrm e}^{{\mathrm i}(P-{\rm d}\Gamma(\hat p))\cdot (X_t-X_s)} \Big]
	\end{equation}
	defines a family of bounded operators satisfying the flow equation
	\begin{equation}\label{eq:flow}
		 S_{r,t}^\Omega(P) S_{s,r}^\Omega (P) = S_{s,t}^\Omega(P),
	\end{equation}
	leaving ${\mathscr D}={\mathscr D}({\rm d}\Gamma(\omega))\cap {\mathscr D}(\Psi(P-{\rm d}\Gamma(\hat p))$ invariant
	and 
	solving the initial value problem
	\begin{align}\label{eq:IVP} \frac{{\rm d}}{{\rm d}t}S_{s,t}^\Omega(P)\psi = -(\Psi(P-{\rm d}\Gamma(\hat p)) + h(t,0))S_{s,t}^\Omega(P)\psi, \quad S_{s,s}^\Omega(P)\psi = \psi, \qquad \psi\in{\mathscr D}.
	\end{align}
\end{prop}
\begin{proof}
By the usual bounds \cref{eq:abound} and the Kato--Rellich theorem, the selfadjointness of $h(t,x)$ follows.  
The remaining proof  is a straightforward extension of methods from \cite{GueneysuMatteMoller.2017,HinrichsMatte.2022} and hence moved to \cref{appendix}.
\end{proof}
\begin{rem}
	By reinspecting our definitions, the reader will recognize that
	\begin{align}\label{def:pathintegralevolution} u_{\alpha^-,\alpha^+}(S,T) = \int_S^T\braket{{\mathrm e}^{-i\hat p\cdot X_t}v^-(t)|U_{v^+_t}(S,t)}{\rm d}t, 
	\end{align}
	an expression which is used as definition for the complex action in previous articles.
\end{rem}
We can now derive a Trotter-type product  formula for different domains of the Nelson-type semigroups introduced above.
In the statement, for $\Omega_1,\Omega_2\subset {\mathbb R}^d$ measurable and $P_1,P_2\in{\mathbb R}^d$, we will use the operator
\begin{align}\label{def:LP}
	L_{P_1,P_2}(\Omega_1,\Omega_2) \coloneqq \Psi(P_1-{\rm d}\Gamma(\hat p\chr_{\Omega_1}))+ \Psi(P_2-{\rm d}\Gamma(\hat p\chr_{\Omega_2})) - \Psi(P_1+P_2-{\rm d}\Gamma(\hat p))
\end{align}
on 
${\mathcal F}(\Omega_1\cup\Omega_2)$.
To ensure that $L_{P_1,P_2}(\Omega_1,\Omega_2)$ is selfadjoint and bounded from below, we
will require the following assumption on the particle dispersion $\Psi$. 
\newcommand{\assu}[1]{$\mathbf{A}(#1)$}
 \begin{assumption}[\assu{\Omega_1,\Omega_2,P_1,P_2}]\hypertarget{eq:subadd}
There exists $C\in{\mathbb R}$ such that 
for all $n\in{\mathbb N}_0$,  
all $k_1,\ldots,k_n\in\Omega_1$ 
and 
all $p_1,\ldots,p_n\in\Omega_2$, it holds that 
\[\Psi\left(P_1-\sum_{i=1}^{n}k_i\right) + \Psi\left(P_2-\sum_{i=1}^{n}p_i\right) - \Psi\left(
P_1+P_2 - \sum_{i=1}^n(k_i+p_i)\right) \ge C.\]
\end{assumption}
\newcommand{\A}[1]{\hyperlink{eq:subadd}{\bf A}($#1$)}

\begin{ex}\label{ex:levsymbolgrowth}
We show examples of $\Psi$ satisfying \assu{\Omega_1,\Omega_2,P_1,P_2}.
\begin{enumerate}
	\item If $\Psi$ is subadditive, like in the case of the semi-relativistic dispersion relation 
	$\Psi(p) = \sqrt{|p|^2+M^2}-M^2$, then the statement easily follows for arbitrary $\Omega_1,\Omega_2\subset {\mathbb R}^d$ 
	and $P_1,P_2\in{\mathbb R}^d$.
	\item Let $\Psi$ be a L\'evy symbol. Then it satisfies $|\Psi(p)| \le C|p|^2$ for some $C>0$, by the L\'evy--Khintchine formula \cref{LL}. 
	Thus, if $\Omega_1\cup \Omega_2$ is bounded, 
	then \assu{\Omega_1,\Omega_2,P_1,P_2} is satisfied
	for arbitrary $P_1,P_2\in{\mathbb R}^d$.
	In particular, in the case of the non-relativistic dispersion relation $\Psi(p)=\frac12p^2$, 
	\assu{\Omega_1,\Omega_2,P_1,P_2} is fulfilled for 
	arbitrary $P_1,P_2\in{\mathbb R}^d$  
	if $\Omega_1\cup \Omega_2$ is bounded. 
\end{enumerate}
\end{ex}

 The next statement is now a main technical ingredient to our proof of ergodicity in the subsequent theorem.
\begin{prop}[Trotter Product Formula]\label{cor:trotter}
	Let $\Omega_1,\Omega_2\subset {\mathbb R}^d$ be disjoint, let $\Theta\subset \Omega\coloneqq \Omega_1\cup\Omega_2$ such that Assumption {\rm \A{\Theta\cap \Omega_1,\Theta\cap \Omega_2,P_1,P_2}} holds for $P_1,P_2\in{\mathbb R}^d$
	and let the assumptions of \cref{prop:evolution} be fulfilled.
	 Further, let $s,T\ge 0$
	and assume that there exists $C>0$ such that 
	\begin{align}\label{eq:limgrowth}
		\sup_{t\in[s,s+T]}\|{\rm Q}^\Omega_\Theta S_{s,t}^\Omega(P_1+P_2){\rm Q}^\Omega_\Theta\|\le {\mathrm e}^{C(t-s)}.
	\end{align}
	For $i=1,2$ and $t\in[s,s+T]$, we define 
	 $S^{\Omega_i}_{s,t}(P)$ as
	 in \cref{def:SOmega}
	 with $\omega$ and $g_{\pm}$ replaced by $\omega\lceil_{\Omega_i}$ and $g_\pm \lceil_{\Omega_i}$, respectively.
	 Then, using the identification 
$
 {\mathcal F}(\Omega) \cong {\mathcal F}(\Omega_1)\otimes 
 {\mathcal F}(\Omega_2)$, 
we have
	\begin{align*}
		{\rm Q}^\Omega_\Theta \left(
		S_{s,s+T}^{\Omega_1}(P_1)\otimes S_{s,s+T}^{\Omega_2}(P_2) \right){\rm Q}^\Omega_\Theta =
		\slim_{N\to\infty} \prod_{i=1}^{N} 
		\left({\rm Q}^\Omega_\Theta S_{s+(i-1)\frac TN,s+i\frac TN}^{\Omega}(P_1+P_2){\rm Q}^\Omega_\Theta {\mathrm e}^{-\frac T{N} L_{P_1,P_2}(\Omega_1\cap\Theta,\Omega_2\cap\Theta)}\right).
	\end{align*}
\end{prop}
\begin{proof}
	We want to apply the Chernoff formula \cite[Theorem, p.~2925]{Vuillermot.2010}.
	To use similar notation to that article, let
\begin{align*}	
&U(s,t) \coloneqq {\rm Q}^\Omega_\Theta 
\left(S_{s,t}^{\Omega_1}(P_1)\otimes S_{s,t}^{\Omega_2}(P_2) \right){\rm Q}^\Omega_\Theta = ({\rm Q}^{\Omega_1}_{\Theta\cap\Omega_1} S_{s,t}^{\Omega_1}(P_1){\rm Q}^{\Omega_1}_{\Theta\cap\Omega_1}) \otimes ({\rm Q}^{\Omega_2}_{\Theta\cap\Omega_2} S_{s,t}^{\Omega_2}(P_2){\rm Q}^{\Omega_2}_{\Theta\cap\Omega_2}),\\
&\xi_t(h) \coloneqq {\rm Q}^\Omega_\Theta S^\Omega_{t,t+h}(P_1+P_2){\rm Q}^\Omega_\Theta{\mathrm e}^{-h L},
\end{align*}
where we abbreviated $L\coloneq L_{P_1,P_2}(\Omega_1\cap\Theta,\Omega_2\cap\Theta)$. 
Note that since $S^\Omega_{s,t}(P)$ leaves 
${\mathcal F}(\Theta)$ invariant, the flow equation \cref{eq:flow} carries over, i.e., 
\[ U(r,t)U(s,r) = U(s,t), \quad r,t\in[s,s+T]. \]
	By the continuity of L\'evy symbols, we see that $\mathscr D\coloneqq {\mathscr D}(L)\cap{\mathscr D}({\rm d}\Gamma(\omega))$ is a core for the selfadjoint operator $\Psi(P_1-{\rm d}\Gamma(\hat p)) + \Psi(P_2-{\rm d}\Gamma(\hat p)) + h(t,0)$, by employing test vectors with only finitely many non-vanishing Fock space components which are smooth functions with bounded support.\footnote{For readers more familiar with Fock space calculus this is just the finite particle subspace over the smooth function with bounded support.}
	
	The main result from \cite[Theorem, p.2925]{Vuillermot.2010} now directly implies 
	  \cref{cor:trotter} if, for any $\psi\in\mathscr D$, we can verify the four assumptions below:
	  \begin{subequations}
	\begin{align}
		& \xi_t(0) = \one,
		\quad t\in[s,s+T],
		\label{eq:Vui1}\\
		& \sup_{t\in[s,s+T]}\|\xi_t(h)\|\le {\mathrm e}^{Ch},
		\label{eq:Vui2}\\
		& \frac{{\rm d}}{{\rm d}t} U(s,t)\psi = -(\Psi(P_1+P_2-\Omega)+h(t,0)+L)U(s,t)\psi,
		\quad t\in[s,s+T]
		\label{eq:Vui3},\\
		& \lim_{\tau\downarrow 0}\sup_{t\in[s,s+T]}\Big\| \underbrace{\Big(\tfrac{1}{\tau}\left(\xi_t(\tau)-1\right) + \Psi(P_1+P_2-{\rm d}\Gamma(\hat p))+h(t,0)+L\Big)}_{\eqqcolon D_t(\tau)} U(s,t)\psi \Big\| = 0.
		\label{eq:Vui4}
	\end{align}
	  \end{subequations}
	We first note that
	\cref{eq:Vui1} is trivial.
	Further, \cref{eq:Vui2} follows from \cref{eq:limgrowth}, since ${\rm Q}^\Omega_\Theta$ is a projection and $L$ is selfadjoint and bounded from below, by Assumption \A{\Omega_1\cap\Theta,\Omega_2\cap\Theta,P_1,P_2}. 
	The differential equation \cref{eq:IVP}, which follows from \cref{prop:evolution}, now implies \cref{eq:Vui3}.
	
	It remains 
	to verify \cref{eq:Vui4}.
	To this end, we first note that for any $\psi\in {\rm Q}^\Omega_\Theta \mathscr D$ and again employing that $\mathcal F(\Theta)$ is left invariant by $S^\Omega_{s,t}(P)$ and $L$, we have
	\begin{align*}
		D_t(\tau)\psi = & \Big(\tfrac1\tau\big(S_{t,t+\tau}^\Omega(P_1+P_2) - 1\big) - \Psi(P_1+P_2-{\rm d}\Gamma(\hat p))+h(t,0) \Big) \psi
		\\
		& + S_{t,t+\tau}^\Omega(P_1+P_2) \Big( \tfrac 1\tau\big({\rm e}^{-\tau L}- 1\big) -  L \Big)\psi
		 + \big(S_{t,t+\tau}^\Omega(P_1+P_2)-1\big)L\psi.
	\end{align*}
	The second and third line are easily observed to converge to zero as $\tau\to0$ uniformly in $t$, by \cref{eq:limgrowth} and the fact that continuous maps on compact sets are uniformly continuous.
	Further, by \cref{eq:IVP}, we have
	\begin{align*}
		\big(S_{t,t+\tau}^\Omega(P_1+P_2) - 1\big)\psi = - \int_0^\tau S_{t,t+r}^\Omega(P_1+P_2)\big(\Psi(P_1+P_2-{\rm d}\Gamma(\hat p))+h(t+r,0)\big)\psi{\rm d}r,
	\end{align*}
	which converges to zero 
	uniformly in $t$ as $\tau\to0$, since 
	the map $(t,x)\mapsto h(t,x)\psi$ is continuous for any $\psi\in{\mathscr D}({\rm d}\Gamma(\omega))$.
	Summarizing the above observations,
	we have shown that 
	\begin{align}\label{eq:uniformconvfixedvector}
		\lim_{\tau\downarrow0}\sup_{t\in[s,s+T]}\norm{D_t(\tau)\psi} = 0, \qquad \psi\in{\mathscr D}.
	\end{align}
	From now fix $\psi\in\mathscr D$ and
	observe that $[s,s+T]\ni t\mapsto U(s,t)\psi \in \mathscr D$ is a continuous map with respect to the graph norm of ${\mathrm d}\Gamma(\omega) + L$.
	Thus $\{U(s,t)\psi|t\in [s,s+T]\}\subset \mathscr D$ is compact with respect to the graph norm and for some $\eps>0$, we can pick a
	finite $\eps$-covering $\{\Phi^\eps_j|j=1,\ldots,N_\eps\}$ of $\{U(s,t)\psi|t\in [s,s+T]\}$. 
	Then for a fixed $\tau>0$, we find
	\begin{align*}
		\sup_{t\in[s,s+T]}\Big\| D_t(\tau) U(s,t)\psi \Big\|
		 \le \max_{j=1,\ldots,N_\eps} \sup_{t\in[s,s+T]} \|D_t(\tau)\Phi^\eps_j\| + \eps \sup_{t\in[s,s+T]}\sup_{\sigma\in(0,T-t]}\| D_t(\sigma)
		 \|_{\mathcal B(\mathscr D,\mathcal F(\Omega))},
	\end{align*}
	where the operator norm in the second term on the right hand side is finite by the uniform boundedness principle.
	Employing \cref{eq:uniformconvfixedvector}, we can now first take the limit $\tau\downarrow0$ and afterwards the limit $\eps\downarrow0$, yielding \cref{eq:Vui4}. This concludes the proof.
\end{proof}
%
Let us for the convenience of the reader note the following simple criterion for the validity of \cref{eq:limgrowth}.
It is inspired by the explicit treatment of the $d=2$ semi-relativistic Nelson model in \cite[Appendix B]{HinrichsMatte.2022}.
\begin{lem}
Suppose that $\|U_{v^-}(s,t)\|_\omega\le C|t-s|^{1/2}$ and $\|U_{v^+_t}(s,t)\|_\omega \le C$.
Then 
 \cref{eq:limgrowth} assumed in 
  \cref{cor:trotter}
 holds true. 
 \end{lem}
\begin{proof}
The statement  can directly be deduced from \cref{eq:Fbd,def:pathintegralevolution}. 
\end{proof}

We can now state our main result on self-energy renormalized models.
The main new ingredient of the proof, compared to the previous theorem, is an application of the above Trotter product formula in the spirit of \cite{Faris.1972} or \cite{Miyao.2018}.
\begin{thm}\label{thm:posselfen}
Let $T>0$, let $P,P_n\in{\mathbb R}^d$ for $n\in{\mathbb N}$, and let $\Omega_1\subset\Omega_2\subset \cdots\subset{\mathbb R}^d$ be an increasing sequence of measurable subsets.
	Set $\Omega= \bigcup_{n=1}^\infty\Omega_n$. 
	%
	For all $n\in{\mathbb N}$, assume that $f^\pm_n\in C^1([0,\infty))$ and that $g^\pm_n\in C([0,\infty],{\mathfrak D}_\omega(\Omega))$ satisfy $g^\pm_n(t)\lceil_{\Omega_n} \in L^2_+(\Omega_n) $ and $g^\pm_m\ge g^\pm_n$ for $m\ge n$ almost everywhere.
	Further set $f^\pm_0(t)=g^\pm_0(t)=0$  
	for all $t\ge 0$.
	For $n\in\mathbb N_0$, we define
	\begin{align}
		\begin{aligned}
			&v_{n,t}^-(s) = v_n^-(s) = {\mathrm e}^{-f_n(s)\omega}g^-_n(s),
			&& v^+_{n,t}(s) = {\mathrm e}^{-|f_n(t)-f_n(s)|\omega}g^+_n(s),\\
			& \alpha^-_n(s) = 2^{-1/2}v^-_n(s),
			&& \alpha^+_n(t,s) = 2^{-1/2}v^+_{n,t}(s).
		\end{aligned}
	\end{align}
	Let $U_v(s,t)$ and $u_{\alpha,\beta}(s,t)$ be given by \cref{def:pathintegral}. 
	Assume that for all $s,t\in[0,T]$, $s\le t$ and $k\in{\mathbb N}_0$ there exists a sequence $(e_n(s,t))_{n\in{\mathbb N}}\subset{\mathbb R}$, a ${\mathbb C}$-valued random variable $u_{\infty,k}(s,t)$, ${\mathfrak D}_\omega(\Omega)$-valued random variables $U^\pm_\infty(s,t)$ and 
	$p,q,r,r'\ge 1$ with $\frac1p+\frac2q=1$ and $\frac1r+\frac1{r'}=1$ such that
	\begin{subequations}\label{eq:convassumstrong}
			\begin{align}
				&\label{eq:convassumstrong.1} {\mathbb E}\left[|{\mathrm e}^{{pu_{\infty,0}}(s,t)}|\right]\le{\mathrm e}^{c(t-s)},\quad {\mathbb E}\left[|{\mathrm e}^{{pu_{\infty,k}}(s,t)}|\right]<\infty,\\ 
				&\label{eq:convassumstrong.2} \sup_{n\in{\mathbb N}}{\mathbb E}\left[{\mathrm e}^{4qr(\|U_{v_{n,t}^\pm}(s,t)\|\vee\|U^\pm_\infty(s,t)\|)^2}\right]<{\mathrm e}^{c(t-s)},\\
				&\label{eq:convassumstrong.3} \sup_{\substack{s,t\in[0,T]\\s\le t}}{\mathbb E}\left[\left|{\mathrm e}^{
					{u_{\alpha_n^--\alpha_k^-,\alpha_n^+-\alpha_k^+}(s,t)+e_n(s,t)-e_k(s,t)}}
				-{\mathrm e}^{u_{\infty,k}(s,t)}\right|^p\right]\xrightarrow{n\to\infty}0,\\
				&\label{eq:convassumstrong.4} \sup_{\substack{s,t\in[0,T]\\s\le t}}{\mathbb E}\left[\|U_{v_{n,t}^\pm}(s,t)-U^\pm_{\infty}(s,t)\|^{qr'}\right]\xrightarrow{n\to\infty}0.
		\end{align}
	\end{subequations}
	Now if
	\begin{align}\label{eq:vacexp}
		{\mathbb E}\left[{\mathrm e}^{u_{\infty,k}(s,t) + {\mathrm i}(P-P_k)\cdot (X_t-X_s)}\right]		 \ne 0 \qquad \mbox{for all}\ k\in{\mathbb N}_0,
	\end{align}
then the ${\mathcal B}({\mathcal F}(\Omega))$-valued Bochner--Lebesgue integral
	\begin{align}\label{def:UVren}
			S^{(\infty,0)}_{s,t}(P)\coloneqq{\mathbb E}\Big[ {\mathrm e}^{u_{\infty,0}(s,t)}F^\omega_{f(t)/2}(U^-_{\infty}(s,t))F^\omega_{f(t)/2}(U^+_{\infty}(s,t)){\mathrm e}^{{\mathrm i}(P-{\rm d}\Gamma(\hat p))\cdot(X_t-X_s)} \Big]
	\end{align} 
	is well-defined and positivity improving for all $s,t\in[0,T]$ such that $s\le t$.
\end{thm}
\begin{proof}
	Throughout this proof, we assume $s,t\in[0,T]$ with $s\le t$, set $\Omega_0=\emptyset$ and for $n\in{\mathbb N}$, $k\in{\mathbb N}_0$ with $n\geq k$ denote
	\begin{align*}
		& \alpha_{n,k}^{\pm} \coloneqq (\alpha_n^\pm - \alpha_k^\pm)\lceil_{\Omega_n\setminus\Omega_k},
		\ v_{n,k,t}^-\coloneqq v_{n,k}^- \coloneqq (v_n^- - v_k^-)\lceil_{\Omega_n\setminus\Omega_k},
		\ v_{n,k,t}^+ \coloneqq (v_{n,t}^+  - v_{k,t}^+)\lceil_{\Omega_n\setminus\Omega_k},
		\\
		&
		S^{(n,k)}_{s,t}(P) \coloneqq {\mathbb E}\Big[ {\mathrm e}^{u_{\alpha_{n,k}^-,\alpha_{n,k}^+}(s,t)+e_n(s,t)-e_k(s,t)}F^\omega_{f(t)/2}(U_{v_{n,k,t}^-}(s,t))F^\omega_{f(t)/2}
		(U_{v_{n,k,t}^+}(s,t))^*{\mathrm e}^{{\mathrm i}(P-{\rm d}\Gamma(\hat p))\cdot (X_t-X_s)} \Big],
	\end{align*}
	where 
		the latter is well-defined in $\mathcal B(\mathcal F(\Omega_n\setminus\Omega_k))$ similar to the argument in Step 1 of the proof of \cref{thm:posren}.
	Further, given a measurable $\Theta\subset \Omega$, we will employ the identifications $L^2(\Omega) \cong L^2(\Theta) \oplus L^2(\Omega\setminus\Theta)$ as well as $\mathcal F(\Omega) \cong \mathcal F(\Theta) \oplus \mathcal F(\Theta)^\perp$ as in \cref{eq:Fockres}, which especially allows us to treat $S^{(n,k)}_{s,t}$ as bounded operators on $\mathcal F(\Omega_n)$ and $\mathcal F(\Omega)$ as well.
	
	Now note that the convergence in \cref{eq:convassumstrong.4} carries over to 
	$v_{n,k,t}^\pm$, i.e., setting $U^\pm_{\infty,k} \coloneqq U^\pm_{\infty}\lceil_{\Omega\setminus\Omega_k}$ we have
\[\sup_{\substack{s,t\in[0,T]\\s\le t}}{\mathbb E}\left[\|U_{v_{n,k,t}^\pm}(s,t)-U^\pm_{\infty,k}(s,t)\|^{qr'}\right]\xrightarrow{n\to\infty}0.\]
	We  write
	\begin{align*}
		S^{(\infty,k)}_{s,t}(P) \coloneqq {\mathbb E}\Big[ {\mathrm e}^{u_{\infty,k}(s,t)}
		F^\omega_{f(t)/2}(U^+_{{\infty,k}}(s,t))F^\omega_{f(t)/2}(U^-_{\infty,k}(s,t))^*{\mathrm e}^{{\mathrm i}(P-{\rm d}\Gamma(\hat p))\cdot (X_t-X_s)} \Big],
	\end{align*}
	which is well defined in $\mathcal B(\mathcal F(\Omega\setminus\Omega_k))$  by the assumptions \cref{eq:convassumstrong.1,eq:convassumstrong.2}.
	
	Further, similar to Step 1 of the proof of \cref{thm:posren}, the assumptions \cref{eq:convassumstrong.1,eq:convassumstrong.2,eq:convassumstrong.3,eq:convassumstrong.4} imply that
	\begin{align}
		\label{eq:uniformconv}
		&\lim_{n\to\infty}\sup_{\substack{s,t\in[0,T]\\s\le t}}\|S^{(\infty,k)}_{s,t}(P) - S^{(n,k)}_{s,t}(P)\|  = 0 
		\qquad\text{for all}\ k\in{\mathbb N}_0. 
	\end{align}
	A similar argument using \cref{eq:convassumstrong.1,eq:convassumstrong.2} and H\"olders inequality yields 
	\begin{align}
			\label{eq:uniformbound}
			\sup_{\substack{n\in{\mathbb N}\cup\{\infty\}}}
			\sup_{\substack{s,t\in[0,T]\\s\le t}}\|S^{(n,0)}_{s,t}(P)\| \le {\rm e}^{c(t-s)}.
	\end{align}
	Finally, note that \cref{thm:regular} as well as our monotonicity and positivity assumptions on $g_n^\pm$ imply that 
	$S^{(k,0)}_{s,t}(P)$
	is positivity improving on ${\mathcal F}(\Omega_k)$ for any $k\in{\mathbb N}$. By \cref{eq:uniformconv}, this yields that $S^{(\infty,k)}_{s,t}(P)$ is positivity preserving on $\mathcal F(\Omega\setminus \Omega_k)$ and thus also on $\mathcal F(\Omega)$ for any $k\in{\mathbb N}_0$.
	
	We now want to apply the Trotter product formula established in \cref{cor:trotter}, but Assumption \A{\Omega_n,\Omega\setminus\Omega_n,P_n,P-P_n)} is not satisfied in general, whence 
	we introduce bounded subsets of $\Omega$  to apply the Trotter product formula. 
	To this end let $\Theta$ be any bounded subset of $\Omega$. 
	Since $\Psi$ is the L\'evy symbol for the stochastic process $X=(X_t)_{t\geq0}$, 
	\cref{ex:levsymbolgrowth} (2) implies that $\Psi$ satisfies Assumption \A{\Theta\cap \Omega_n,\Theta\cap (\Omega\setminus\Omega_n),P_n,P-P_n)}. 
	Since we work under the assumptions of \cref{prop:evolution} with $g_\pm =g_n^\pm$ and $f=f_n^\pm$ and since we verified assumption \cref{eq:limgrowth} by \cref{eq:uniformbound},
	\cref{cor:trotter} and the identification ${\mathcal F}(\Omega)\cong{\mathcal F}(\Omega_k)\otimes 
	{\mathcal F}(\Omega\setminus\Omega_k)$ imply that
	\begin{align}\label{c}
		{\rm Q}^\Omega_\Theta \left(
		S^{(k,0)}_{s,t}(P_k)\otimes S^{(n,k)}_{s,t}(P-P_k) \right)
		{\rm Q}^\Omega_\Theta 
		= \slim_{N\to\infty} \prod_{i=1}^{N} \left({\rm Q}^\Omega_\Theta S_{s+(i-1)\frac{t-s}N,s+i\frac {t-s}N}^{(n,0)}(P){\rm Q}^\Omega_\Theta  {\mathrm e}^{-\frac tN L_k}\right)
	\end{align}
	for all $n\in{\mathbb N}$ and $k\in {\mathbb N}_0$,
	 where we defined $L_k\coloneqq L_{P_k,P-P_k}(\Theta\cap \Omega_k,\Theta \cap (\Omega\setminus \Omega_k))$.
	From \cref{eq:uniformconv} we immediately see that the left hand side of \cref{c} converges in norm as $n\to\infty$.
	Furthermore, for a fixed $N\in{\mathbb N}$, we see that 
	$\prod_{i=1}^{N} \left({\rm Q}^\Omega_\Theta S_{s+(i-1)\frac{t-s}N,s+i\frac {t-s}N}^{(n,0)}(P){\rm Q}^\Omega_\Theta  {\mathrm e}^{-\frac tN L_k}\right)$ 
 converges in norm as $n\to\infty$ and the convergence is uniform in $N\in{\mathbb N}$ by \cref{eq:uniformbound}.
	Thus, for any $k\in{\mathbb N}_0$, 
	\begin{align}\label{eq:trottren}
		{\rm Q}^\Omega_\Theta \left(
		S^{(k,0)}_{s,t}(P_k)\otimes S^{(\infty,k)}_{s,t}(P-P_k) \right)
		{\rm Q}^\Omega_\Theta 
		= \slim_{N\to\infty} \prod_{i=1}^{N} \left({\rm Q}^\Omega_\Theta S_{s+(i-1)\frac{t-s}N,s+i\frac {t-s}N}^{(\infty,0)}(P){\rm Q}^\Omega_\Theta  {\mathrm e}^{-\frac tN L_k}\right).
	\end{align}
	We now prove that
	${\rm Q}^\Omega_\Theta S^{(\infty,0)}_{s,t}(P){\rm Q}^\Omega_\Theta\lceil_{\mathcal F(\Theta)}$ is in fact positivity improving on $\mathcal F(\Theta)$.
	Thus, let us assume we have $\phi,\psi \in \FBP(\Theta)$ such that
	 $\braket{{\rm Q}^\Omega_\Theta S_{s,t}^{(\infty,0)}(P) {\rm Q}^\Omega_\Theta \phi,\psi}=
	 \braket{S_{s,t}^{(\infty,0)}(P)\phi,\psi} = 0$.
		Since $S^{(\infty,0)}_{s,t}(P)$ is positivity preserving and thus its adjoint is positivity preserving as well, cf. \cite{Miura.2003} or \cite[Lemma~2.7]{Miyao.2018},\footnote{In fact, $S^\Omega_{s,t}(P)$ is selfadjoint, but to prove this from the functional integral representation one needs to employ a rather complicated technical procedure, cf. \cite{MatteMoller.2018}. Thus, we refrain from doing so here and use a simpler abstract result.} we know
	$
		\supp({\rm P}\!_m \phi)\cap \supp({\rm P}\!_mS^{(\infty,0)}_{s,t}(P)^*\psi) 
	$
	is a Lebesgue zero set for any $m\in{\mathbb N}$.
	Further,
	since $L_k$ is a multiplication operator on ${\mathcal F}^{(m)}(\Theta)$ the same holds for
	$
	\supp({\rm P}\!_m {\mathrm e}^{-qL_k}\phi)\cap \supp({\rm P}\!_mS^{(\infty,0)}_{s,t}(P)^\ast\psi) 
	$ for any $q\geq0$,
	so
	\begin{align*}
		\Braket{ {\mathrm e}^{-\frac {t-s}NL_k}\phi,S^{(\infty,0)}_{s,t}(P)^*\psi} = 0.
	\end{align*}
By the flow equation \cref{eq:flow} we have 
$S^{(\infty,0)}_{s,t}(P)=S^{(\infty,0)}_{s+\frac{t-s}{N},t}(P)S^{(\infty,0)}_{s,s+\frac{t-s}{N}}(P)$ and thus 
	\begin{align*}
		\Braket{{\rm Q}^\Omega_\Theta S^{(\infty,0)}_{s,s+\frac{t-s}N}(P){\rm Q}^\Omega_\Theta {\mathrm e}^{-\frac {t-s}NL_k}\phi,S^{(\infty,0)}_{s+\frac{t-s}{N},t}(P)^*\psi} = 0.
	\end{align*}
Similarly we obtain that 
	\begin{align*}
		\Braket{{\mathrm e}^{-\frac {t-s}NL_k}{\rm Q}^\Omega_\Theta S^{(\infty,0)}_{s,s+\frac{t-s}N}(P){\rm Q}^\Omega_\Theta {\mathrm e}^{-\frac {t-s}NL_k}\phi,
		S^{(\infty,0)}_{s+\frac{t-s}{N},t}(P)^*\psi} = 0.
	\end{align*}
Again by the flow equation \cref{eq:flow}, 
we have $S^{(\infty,0)}_{s+\frac{t-s}{N},t}(P)=S^{(\infty,0)}_{s+\frac{2(t-s)}{N},t}(P)
S^{(\infty,0)}_{s+\frac{t-s}{N},s+\frac{2(t-s)}{N}}(P)$ and 
\begin{align*}
		\Braket{
{\rm Q}^\Omega_\Theta	S^{(\infty,0)}_{s+\frac{t-s}{N},s+\frac{2(t-s)}{N}}(P){\rm Q}^\Omega_\Theta
		{\mathrm e}^{-\frac {t-s}NL_k}{\rm Q}^\Omega_\Theta S^{(\infty,0)}_{s,s+\frac{t-s}N}(P){\rm Q}^\Omega_\Theta {\mathrm e}^{-\frac {t-s}NL_k}\phi,
		S^{(\infty,0)}_{s+\frac{2(t-s)}{N},t}(P)^*\psi} = 0.
	\end{align*}
	Iterating this procedure, we arrive at
	\begin{align*}
		\Braket{\left(\prod_{\ell=1}^{N} {\rm Q}^\Omega_\Theta S^{(\infty,0)}_{s+\frac{\ell-1}{N}(t-s),s+\frac\ell N (t-s)}(P){\rm Q}^\Omega_\Theta{\mathrm e}^{-\frac {t-s}NL_k}\right) \phi, \psi} = 0.
	\end{align*}
	In the limit $N\to\infty$, the Trotter product formula \cref{eq:trottren} now implies
	\begin{align*}
		\Braket{S^{(k,0)}_{s,t}(P_k)\otimes S^{(\infty,k)}_{s,t}(P-P_k) \phi, \psi} = 0.
	\end{align*}
	We set
 ${\rm Q}_k={\rm Q}^\Theta_{\Theta\cap \Omega_k}$
and rewrite the above identity as
	\begin{align*}
		0=&\Braket{S^{(k,0)}_{s,t}(P_k)\otimes S^{(\infty,k)}_{s,t}(P-P_k) (\one-{\rm Q}_k)\phi, (\one-{\rm Q}_k)\psi}
		\\&+\Braket{S^{(k,0)}_{s,t}(P_k)\otimes S^{(\infty,k)}_{s,t}(P-P_k) {\rm Q}_k\phi, (\one-{\rm Q}_k)\psi}\\
		&+\Braket{S^{(k,0)}_{s,t}(P_k)\otimes S^{(\infty,k)}_{s,t}(P-P_k) (\one-{\rm Q}_k)\phi, {\rm Q}_k\psi}
		+\Braket{S^{(k,0)}_{s,t}(P_k)\otimes S^{(\infty,k)}_{s,t}(P-P_k) {\rm Q}_k
		\phi, {\rm Q}_k\psi}. 		
	\end{align*}
Now using that
$(\one-{\rm Q}_k)\phi, 
(\one-{\rm Q}_k)\psi,
{\rm Q}_k\phi, 
{\rm Q}_k\psi 
\in \FBP(\Theta)\subset \FBP(\Omega)$,
by \cref{24,25},
all terms in the sum are non-negative, whence 
we have
\[\Braket{S^{(k,0)}_{s,t}(P_k)\otimes S^{(\infty,k)}_{s,t}(P-P_k) {\rm Q}_k
		\phi, {\rm Q}_k\psi}=0.\]
Under the identification
$\mathcal F(\Theta)\cong \mathcal F(\Theta\cap\Omega_k)\otimes\mathcal F(\Theta\setminus \Omega_k)$, the vector
${\rm Q}_k
		\phi$ is of the form
${\rm Q}_k
		\phi	=\phi_k\otimes \Omega_F$ with $\phi_k\in \overline{{\mathcal F}}_+(\Theta)$,
where $\Omega_F$ is the Fock vacuum of 
${\mathcal F}(\Theta\setminus \Omega_k)$. 
		Similarly 
${\rm Q}_k
		\psi	=\psi_k\otimes \Omega_F$ with $\psi_k\in \overline{{\mathcal F}}_+(\Theta)$. 
This yields
\begin{align*}
0&=\Braket{S^{(k,0)}_{s,t}(P_k)\otimes S^{(\infty,k)}_{s,t}(P-P_k) {\rm Q}_k
		\phi, {\rm Q}_k\psi}=
		\Braket{S^{(k,0)}_{s,t}(P_k)\phi_k,
		 \psi_k}
		\Braket{\Omega_F,  S^{(\infty,k)}_{s,t}(P-P_k) \Omega_F}\\
		&=
\Braket{S^{(k,0)}_{s,t}(P_k)\phi_k,
		 \psi_k}{\mathbb E}[{\mathrm e}^{u_{\infty,k} + {\mathrm i}(P-P_k)\cdot(X_t-X_s)}]
		 \end{align*}
and 	hence 	using \cref{eq:vacexp}
	\begin{align*}
		\Braket{S^{(k,0)}_{s,t}(P_k)\phi_k,\psi_k}  = 0.
	\end{align*}
	Thus, since $S^{(k,0)}_{s,t}(P_k)$ is positivity improving on $\mathcal F(\Omega_k)$, we have that  
	$\phi_k = 0$ or $\psi_k =0$ for all $k\in{\mathbb N}$, so $\phi = 0$ or $\psi = 0$. 
	This proves that ${\rm Q}^\Omega_\Theta S^{(\infty,0)}_{s,t}(P){\rm Q}^\Omega_\Theta\lceil \mathcal F(\Theta)$ is indeed positivity improving. Since $\Theta$ can be any arbitrary bounded set, this proves the statement.
\end{proof}
\begin{rem}
	One major technical difficulty of the proof of \cref{thm:posselfen} is the restriction to the bounded subset $\Theta\subset \Omega$, which is due to the fact that Assumption \assu{\Omega_1,\Omega_2,P_1,P_2} is not satisfied for arbitrary sets $\Omega_1,\Omega_2\subset \mathbb R^d$ for $\Psi = \Psi_{\nr}(p)=\half |p|^2$
	An appropriate generalization of the Trotter product formula in \cref{cor:trotter} allowing for $L_{P_1,P_2}(\Omega_1,\Omega_2)$ to be unbounded from below would allow a more direct approach along the lines of our proof, since we could then set $\Theta = \Omega$.
\end{rem}
We conclude with the
\begin{proof}[\textbf{Proof of \cref{thm:Nelsonpos}}]
From \cref{eq:FK,eq:FKUVreg,eq:UVren}, we recall that
	\begin{align*}
		{\mathrm e}^{-tH_{\#,n}(P)}& = {\mathbb E}\Big[{\mathrm e}^{u_{n,t}^\#}F^\omega_{t/2}(U_{n,t}^{\#,-})F^\omega_{t/2}(U_{n,t}^{\#,+})^*{\mathrm e}^{{\mathrm i}(P-{\rm d}\Gamma(\hat p))\cdot X^\#_t}\Big]\\\ 
		&\xrightarrow{n\to\infty}
{\mathbb E}\Big[{\mathrm e}^{u_t^\#}F^\omega_{t/2}(U_t^{\#,-})F^\omega_{t/2}(U_t^{\#,+})^*{\mathrm e}^{{\mathrm i}(P-{\rm d}\Gamma(\hat p))\cdot X^\#_t}\Big]=
{\mathrm e}^{-tH_\#(P)} \end{align*}
in the strong sense. 
Again, we could have also deduced the convergence from \cref{eq:uniformconv}.
We  now apply \cref{thm:posselfen} with $s=0$ to show that 
the right hand side above 
is positivity improving.
	To this end, we choose $\Omega_n=\{k\in{\mathbb R}^d\mid |k|<n\}$, $P_n=P$, $f_n(t)=t$, $g_n^\pm(t) = -v_n=-\lambda \one_{|\cdot|<n}\om(\cdot)^{-1/2}$ and $e_n(s,t) = E_n|s-t|$ for any $s,t\ge 0$.
	Then our definitions $u_{n,t}$ and $U_{n,t}^\pm$ agree with $u_{\alpha_n^-,\alpha_n^+}(0,t)$ and $U_{-v_{n,t}^\pm}(0,t)$, respectively, cf. \cref{eq:uUV,eq:UUV,def:pathintegral}. Further, the positivity and monotonicity assumptions on $g_n^\pm$ are easily verified from the definitions and the assumption $\lambda<0$. We have already recalled the convergence assumptions \cref{eq:convassumstrong.3,eq:convassumstrong.4} in \cref{eq:uconvexample,eq:Uconvexample}  (note that $u_{n,k} = u_n-u_k$ here) and the exponential moment bounds  assumed in \cref{eq:convassumstrong.1,eq:convassumstrong.2} can be found in \cite[Corollary~3.21,~Theorem~4.9]{MatteMoller.2018} for the case $\#=\nr $ and \cite[Theorems~6.6 and 6.7,~Lemma~B.1]{HinrichsMatte.2022} for the case $\#=\sr$.
	The final assumption \cref{eq:vacexp} is trivial in our case, since $u_{n,t}$ is easily seen to be real-valued for any $n\in {\mathbb N}$ and thus also $u_t$ is.
	Hence, all assumptions of \cref{thm:posselfen} are satisfied here and the statement follows.
\end{proof}

\subsection*{Acknowledgements}
The authors thank Oliver Matte for valuable discussions on the subject.
BH acknowledges funding by the Ministry of Culture and Science of the State of North Rhine-Westphalia within the project `PhoQC' (Grant Nr. PROFILNRW-2020-067).
FH is financially supported by JSPS KAKENHI 20K20886 %
and JSPS KAKENHI 20H01808.

\appendix
\section{Evolution Equations for Feynman--Kac Evolution Systems}\label{appendix}

The following proof is a generalization of previous discussions given in \cite{GueneysuMatteMoller.2017,MatteMoller.2018,HinrichsMatte.2022,HinrichsMatte.2023}. Some technical discussions are deferred to those articles.

\begin{proof}[\textbf{Proof of \cref{prop:evolution}}]
	\ 
	
	\smallskip
	\noindent{\em Step 1.} We start by studying the integrand of the path integral
	\begin{align}
		W_{s,t} \coloneqq {\mathrm e}^{u_{\alpha^-,\alpha^+}(s,t)}F_{f(t)/2}(U_{v^+_t}(s,t))F_{f(t)/2}(U_{v^-}(s,t))^*.
	\end{align}
	and prove the integral equation
	\begin{align}\label{eq:intequW}
		W_{s,t}\phi - \phi = -\int_s^t h(r,X_r)W_{s,r}\phi{\rm d}r.
	\end{align}
	To this end, we will first choose $\phi$ to be an exponential vector, i.e., for $h\in L^2(\Omega)$ the Fock space vector given by
	\begin{align}\label{def:expvec}
		{\rm P}\!_n \eps(h) \coloneqq \frac 1{\sqrt{n!}}h^{\otimes n}  \in {\mathcal F}(\Omega).
	\end{align}
	They satisfy
	\begin{align}\label{prop:expvec}
		a(f)\eps(h) = \braket{f,h}\eps(h) \quad\text{and}\quad \braket{\eps(h_1)|{\rm d}\Gamma(\omega)\eps(h_2)} = \braket{h_1|\omega h_2}{\mathrm e}^{\braket{h_1,h_2}}
	\end{align}
for 
$f,h_1\in L^2(\Omega)$ and $ h_2\in{\mathscr D}(\omega)$. 
A direct calculation, cf. \cite[Remark 5.2]{MatteMoller.2018}, yields
	\begin{align*}
		F_t^\omega(f)\eps(h) = \eps(f+{\mathrm e}^{-t\omega}h)
		\quad\text{and}\quad
		F_t^\omega(f)^*\eps(h) = {\mathrm e}^{\braket{f|h}}\eps({\mathrm e}^{-t\omega}h).
	\end{align*}
	Thus
	\begin{align*}
		W_{s,t}\eps(h) = {\mathrm e}^{u_{\alpha^-,\alpha^+}(s,t) - \braket{U_{v^-}(s,t)|h}} \eps ({\mathrm e}^{-f(t)\omega}h - U_{v^+_t}(s,t)).
	\end{align*}
	Inserting our choices \cref{eq:evolutionassum}
	\begin{align*}
		\braket{\eps(h_2)|W_{s,t}\eps(h_1)} = 
		{\mathrm e}^{\int_s^t\! \theta(s,r){\rm d}r + \braket{h_2|h_1}} = \braket{\eps(h_2)|\eps(h_1)} + \int_s^t \theta(s,r)\braket{\eps(h_2)|W_{s,r}\eps(h_1)}{\rm d}r
	\end{align*}
	with
	\begin{align*}
		\theta(s,t) = \braket{{\mathrm e}^{-i\hat p\cdot X_t}g^-(t)| U_{v^+_t}(s,t)}  - \braket{{\mathrm e}^{-{\mathrm i}\hat p \cdot X_t}v^-(t)|h_1} - \braket{h_2|f'(t)\omega({\mathrm e}^{-f(t)\omega}h_1 - U_{v^+_t}(s,t)) + {\mathrm e}^{-{\mathrm i}\hat p X_t}v_t^+(t)}.
	\end{align*}
	In view of \cref{prop:expvec}, this yields
	\begin{align*}
		\braket{\eps(h_2)| W_{s,t}\eps(h_1)} - \braket{\eps(h_2)|\eps(h_1)} = -\int_s^t \braket{\eps(h_2)|h(r,X_r)W_{s,r}\eps(h_1)}{\rm d}r.
	\end{align*}
	Since the span of exponential vectors is dense in $L^2(\Omega)$ and a core for $h(r,X_r)$, by the Kato--Rellich theorem, cf. \cite{Parthasarathy.1992}, this proves \cref{eq:intequW}.
	
	\smallskip
	\noindent{\em Step 2.} We now prove \cref{eq:IVP}, following the proofs of \cite[Lemma~7.1,~Theorem ~7.3]{HinrichsMatte.2023}. For more details, we refer the reader there
	
	To this end, note that for $\eta\in{\mathscr D}(|{\rm d}\Gamma(\hat p)|^2)$, the map $z\mapsto {\mathrm e}^{{\mathrm i}(P-{\rm d}\Gamma(\hat p)\cdot z)}\eta$ is twice differentiable and we can apply It\^o's formula (cf. \cite[Lemma~3.6]{HinrichsMatte.2022} for the version used here) in conjunction with It\^o's product formula and Step 1 to obtain
	\begin{align*}
		&\Braket{{\mathrm e}^{{\mathrm i}(P-{\rm d}\Gamma(\hat p))\cdot X_t}\eta | W_{s,t}\phi} - \braket{\eta|\phi}\\
		& \qquad = -\int_s^t \Braket{{\mathrm e}^{{\mathrm i}(P-{\rm d}\Gamma(\hat p))\cdot X_r}\eta | \Big(\Psi(P-{\rm d}\Gamma(\hat p)) + h(r,X_r)\Big)W_{s,r}\phi}{\rm d}r\\
		& \qquad \qquad + \int_s^t \Braket{{\mathrm e}^{{\mathrm i}(P-{\rm d}\Gamma(\hat p))\cdot X_{r-}}{\rm d}\Gamma(\hat p) \eta| W_{s,r}\phi} \cdot {\rm d}B^\Psi_{r}\\
		& \qquad \qquad + \int_{(s,t]\times {\mathbb R}^d} \Braket{{\mathrm e}^{{\mathrm i}(P-{\rm d}\Gamma(\hat p))\cdot X_{r-}}\left({\mathrm e}^{{\mathrm i}(P-{\rm d}\Gamma(\hat p))\cdot z} - 1 \right)\eta | W_{s,r-}\phi }{\mathrm d}\tilde N_\Psi(r,z),
	\end{align*}
	where $B^\Psi_r$ is the Brownian part of $X$ and $\tilde N_\Psi(r,z)$ is the Poisson random measure associated with $X$, by the L\'evy--It\^o decomposition, cf. \cite[Theorem 2.4.16]{Applebaum.2009}.
	The terms in the last two lines are martingales starting at zero and hence drop out when taking the expectation.
	Further,
	recalling that $S^\Omega_{s,t}(P) = {\mathbb E}\big[W_{s,t}^* {\mathrm e}^{{\mathrm i}(P-{\rm d}\Gamma(\hat p))\cdot X_t}\big]$
	 as well as ${\mathrm e}^{-{\mathrm i}(P-{\rm d}\Gamma(\hat p))\cdot X_t}h(t,X_t) = h(t,0){\mathrm e}^{-{\mathrm i}(P-{\rm d}\Gamma(\hat p))\cdot X_t} $, we find
	\begin{align*}
		S^\Omega_{s,t_2}(P)\phi - S_{s,t_1}\phi = -\int_{t_1}^{t_2} S^\Omega_{s,r}(P) \left(\Psi(P-{\rm d}\Gamma(\hat p)) + h(t,0)\right)\phi {\rm d}r
	\end{align*}
	Selfadjointness of $\Psi(P-{\rm d}\Gamma(\hat p)) + h(t,0)$ and continuity of $t\mapsto S^\Omega_{s,t}(P)$ imply the desired differential equation \cref{eq:IVP}. 
	
	\smallskip
	\noindent{\em Step 3.} It remains to prove \cref{eq:flow}.
	
	Directly from our definitions of $u$ and $U$, it is straightforward to verify
	\begin{align*}
		{\mathrm e}^{-{\mathrm i}(P-{\rm d}\Gamma(\hat p))(X_t-X_s)}W_{s,t} = {\mathrm e}^{-{\mathrm i}(P-{\rm d}\Gamma(\hat p))\cdot (X_t-X_r)}W_{r,t} {\mathrm e}^{-{\mathrm i}(P-{\rm d}\Gamma(\hat p))\cdot (X_r-X_s)} W_{s,r},
	\end{align*}
	by first calculating the action of both sides on exponential vectors and then using a density argument,
	see for example \cite[Lemma~4.18]{MatteMoller.2018} for more details on the explicit calculations.
	By taking expectations and using the stationary and independent increments of our L\'evy process, one easily arrives at \cref{eq:flow}, see for example \cite[Theorem ~6.6 and Corollary 6.7]{HinrichsMatte.2023} for an in-depth discussion.
%
\end{proof}


\bibliographystyle{halpha-abbrv}
\bibliography{../../Literature/00lit}

\end{document}

%% file: header.tex
\usepackage[utf8]{inputenc}
\usepackage[english]{babel}
\usepackage[T1]{fontenc}

\usepackage{amsmath,amssymb,mathtools,bm}
\usepackage{csquotes}
\usepackage{amstext}
\usepackage{amsthm}
\usepackage{bbm}
\usepackage{enumerate}
\usepackage{hyperref}
\usepackage[nameinlink,capitalise]{cleveref}
\usepackage{braket}
\usepackage{dsfont}
\usepackage{color}
\usepackage{tikz}
\usepackage{pgfplots}

\makeatletter
\def\@seccntformat#1{%
	\protect\textup{\protect\@secnumfont
		\ifnum\pdfstrcmp{subsection}{#1}=0 \bfseries\fi
		\ifnum\pdfstrcmp{subsubsection}{#1}=0 \itshape\fi
		\csname the#1\endcsname
		\protect\@secnumpunct
	}%
}
\renewcommand{\@upn}{}
\DeclareRobustCommand{\crefnosort}[1]{%
	\begingroup\@cref@sortfalse\cref{#1}\endgroup
}
\makeatother

\numberwithin{equation}{section}
\newtheorem{thm}{Theorem}
\newtheorem{lem}[thm]{Lemma}
\newtheorem{prop}[thm]{Proposition}
\newtheorem{cor}[thm]{Corollary}

\usepackage{thmtools}
\declaretheoremstyle[
notefont=\bfseries, notebraces={}{},
bodyfont=\normalfont\itshape,
headformat=\NAME \NOTE
]{nopar}
\declaretheorem[style=nopar]{assumption}

\theoremstyle{definition}

\renewcommand*{\thehyp}{\Alph{hyp}}

\theoremstyle{remark}
\newtheorem{rem}{Remark}[section]
\newtheorem{ex}[rem]{Example}

\crefname{hyp}{Hypothesis}{Hypotheses}
\Crefname{hyp}{Hypothesis}{Hypotheses}
\crefname{lem}{Lemma}{Lemmas}
\Crefname{lem}{Lemma}{Lemmas}
\crefname{thm}{Theorem}{Theorems}
\Crefname{thm}{Theorem}{Theorems}
\crefname{prop}{Proposition}{Propositions}
\Crefname{prop}{Proposition}{Propositions}
\crefname{enumi}{}{}
\Crefname{enumi}{}{}
\creflabelformat{enumi}{#2(#1)#3}
\crefname{equation}{}{}
\Crefname{equation}{}{}
\crefname{rem}{Remark}{Remarks}
\Crefname{rem}{Remark}{Remarks}
\crefname{ex}{Example}{Examples}
\Crefname{ex}{Example}{Examples}

\makeatletter
\renewcommand{\@upn}{} 
\makeatother

\usepackage[inline]{enumitem}
\newlist{enumthm}{enumerate}{1} 
\setlist[enumthm]{label=\upshape(\roman*),ref=\thethm~(\roman*)}  
\crefalias{enumthmi}{thm} 
\newlist{enumcor}{enumerate}{1}
\setlist[enumcor]{label=\upshape(\roman*),ref=\thecor~(\roman*)}
\crefalias{enumcori}{cor}
\newlist{enumlem}{enumerate}{1}
\setlist[enumlem]{label=\upshape(\roman*),ref=\thelem~(\roman*)}
\crefalias{enumlemi}{lem}
\newlist{enumprop}{enumerate}{1}
\setlist[enumprop]{label=\upshape(\roman*),ref=\theprop~(\roman*)}
\crefalias{enumpropi}{prop}
\newlist{enumhyp}{enumerate}{1}
\setlist[enumhyp]{label=\upshape(\roman*),ref=\thehyp~(\roman*)}
\crefalias{enumhypi}{hyp}
\newlist{enumproof}{enumerate*}{1}
\setlist[enumproof]{label=\upshape(\roman*)}
\newlist{enumdef}{enumerate}{1}
\setlist[enumdef]{label=\upshape(\roman*),ref=\thedefn~(\roman*)}
\crefalias{enumdefi}{defn}

\makeatletter
\newcounter{subcreftmpcnt} %
\newcommand\romansubformat[1]{(\roman{#1})} 
\def\subcref{\@ifstar\@@subcref\@subcref}
\newcommand\@subcref[2][\romansubformat]{%
	\ifcsname r@#2@cref\endcsname
	\cref@getcounter {#2}{\mylabel}%
	\setcounter{subcreftmpcnt}{\mylabel}%
	\hyperref[#2]{\romansubformat{subcreftmpcnt}}%
	\else ?? \fi}   
\newcommand\@@subcref[2][\romansubformat]{%
	\ifcsname r@#2@cref\endcsname
	\cref@getcounter {#2}{\mylabel}%
	\setcounter{subcreftmpcnt}{\mylabel}%
	\romansubformat{subcreftmpcnt}%
	\else ?? \fi}   
\makeatother

\makeatletter
\DeclareRobustCommand{\crefnosort}[1]{%
	\begingroup\@cref@sortfalse\cref{#1}\endgroup
}
\makeatother

\def\endstepsymbol{$\lozenge$}
\def\endclaimsymbol{$\lozenge$}
\newcounter{proofstep}
\AtBeginEnvironment{proof}{\setcounter{proofstep}{0}}

\crefname{proofstep}{Step}{Steps}
\Crefname{proofstep}{Step}{Steps}
\newcounter{proofclaim}
\AtBeginEnvironment{proof}{\setcounter{proofclaim}{0}}

\crefname{proofclaim}{Claim}{Claims}
\Crefname{proofclaim}{Claim}{Claims}

\newcommand{\RR}{\BR}

\newcommand{\eps}{\varepsilon}\newcommand{\ph}{\varphi}

\newcommand{\supp}{\operatorname{supp}}

\DeclareMathOperator*{\slim}{s-lim}

\renewcommand{\bar}[1]{\overline{#1}}

\DeclareFontFamily{U}{mathx}{\hyphenchar\font45}
\DeclareFontShape{U}{mathx}{m}{n}{
	<5> <6> <7> <8> <9> <10>
	<10.95> <12> <14.4> <17.28> <20.74> <24.88>
	mathx10
}{}
\DeclareSymbolFont{mathx}{U}{mathx}{m}{n}
\DeclareFontSubstitution{U}{mathx}{m}{n}
\DeclareMathAccent{\widecheck}{0}{mathx}{"71}
\DeclareMathAccent{\wideparen}{0}{mathx}{"75}

\DeclareFontFamily{OMX}{MnSymbolE}{}
\DeclareFontShape{OMX}{MnSymbolE}{m}{n}{
	<-6>  MnSymbolE5
	<6-7>  MnSymbolE6
	<7-8>  MnSymbolE7
	<8-9>  MnSymbolE8
	<9-10> MnSymbolE9
	<10-12> MnSymbolE10
	<12->   MnSymbolE12}{}
\DeclareSymbolFont{mnlargesymbols}{OMX}{MnSymbolE}{m}{n}
\SetSymbolFont{mnlargesymbols}{bold}{OMX}{MnSymbolE}{b}{n}
\DeclareMathDelimiter{\llangle}{\mathopen}{mnlargesymbols}{'164}{mnlargesymbols}{'164}
\DeclareMathDelimiter{\rrangle}{\mathclose}{mnlargesymbols}{'171}{mnlargesymbols}{'171}
\DeclareMathDelimiter{\lsem}{\mathopen}{mnlargesymbols}{'102}{mnlargesymbols}{'102}
\DeclareMathDelimiter{\rsem}{\mathclose}{mnlargesymbols}{'107}{mnlargesymbols}{'107}
\DeclareMathDelimiter{\langlebar}{\mathopen}{mnlargesymbols}{'152}{mnlargesymbols}{'152}
\DeclareMathDelimiter{\ranglebar}{\mathclose}{mnlargesymbols}{'157}{mnlargesymbols}{'157}
\DeclareMathDelimiter{\lWavy}{\mathopen}{mnlargesymbols}{'137}{mnlargesymbols}{'137}
\DeclareMathDelimiter{\rWavy}{\mathopen}{mnlargesymbols}{'137}{mnlargesymbols}{'137}

\newcommand{\chr}{\mathbf 1}
\newcommand{\abs}[1]{\lvert#1\lvert}
\newcommand{\norm}[1]{\lVert#1\lVert}

\newcommand{\ad}{a^\dagger}